\documentclass[a4paper,11pt]{article}
\usepackage{jheppub} % for details on the use of the package, please see the JINST-author-manual
\usepackage{cleveref}
\usepackage[T1]{fontenc}
\usepackage[utf8]{inputenc} % inutile avec lualatex/xelatex, conservé pour pdflatex
\usepackage[english]{babel}
\usepackage{microtype}
\usepackage{amsmath,amssymb,amsfonts,amsthm,mathtools}
\usepackage{mathrsfs}
\usepackage{bm}

\usepackage{graphicx}
\usepackage{xcolor}
\usepackage{fancyhdr}
\usepackage{setspace,bbm}
\usepackage[numbers,sort&compress]{natbib}

\theoremstyle{plain}
\newtheorem{theorem}{Theorem}[section]

\theoremstyle{definition}

\theoremstyle{remark}
\newtheorem{remark}[theorem]{Remark}

\numberwithin{equation}{section}

\newcommand{\bR}{\mathbb{R}}
\newcommand{\bC}{\mathbb{C}}

\newcommand{\al}{\alpha}
\newcommand{\ti}{\tilde}
\newcommand{\one}{\mathbbm{1}}

\newcommand{\D}{\mathcal{D}}
\newcommand{\be}{\begin{equation}}
\newcommand{\ee}{\end{equation}}
\newcommand{\bre}{\begin{remark}}
\newcommand{\ere}{\end{remark}}
 \newcommand{\ga}{\gamma}
\newcommand{\jp}{\frac{1}{2}}
\newcommand{\bz}{{\bar z}}
\newcommand{\js}{\frac{1}{4}}
\newcommand{\lm}{\lambda}

\newcommand{\ups}{\Upsilon}
\newcommand{\bpm}{\begin{pmatrix}}
\newcommand{\epm}{\end{pmatrix}}

\newcommand{\E}{\mathcal{E}}

\newcommand{\G}{\mathcal{G}}

\newcommand{\K}{\mathcal{K}}

\newcommand{\M}{\mathcal{M}}

\newcommand{\cP}{\mathcal{P}}
\newcommand{\cS}{\mathcal{S}}
\newcommand{\si}{\sigma}

\newcommand{\tr}{\mathrm{tr}}
\newcommand{\Ad}{\mathrm{Ad}}
\newcommand{\ad}{\mathrm{ad}}

\newcommand{\ri}{\mathrm{i}}

\newcommand{\pa}{\partial}
\newcommand{\xp}{\partial_+}
\newcommand{\xm}{\partial_-}
\newcommand{\xz}{\partial_z}
\newcommand{\xbz}{\partial_\bz}

\arxivnumber{2603.21355} % if you have one

\title{\boldmath Euclidean $\mathcal{E}$-models}

\date{}

\author{Ctirad Klim\v{c}\'{\i}k}
\affiliation{Institut de Math\'ematiques de Marseille, UMR 7373\\ Aix Marseille Université, CNRS,\\
13453 Marseille, France}

\emailAdd{ctirad.klimcik@univ-amu.fr}

\abstract{ 
We develop the theory of the first-order dynamical systems, called the Euclidean $\E$-models, which  naturally give rise in their second-order formulation to
non-unitary nonlinear $\sigma$-models with real Euclidean actions. We
 establish the Euclidean version
of Poisson--Lie T-duality, formulate sufficient conditions for Lax
integrability, and describe the one-loop renormalization flow of the
Euclidean $\mathcal E$-operator.

For perfect Drinfeld doubles, we introduce the $\mathcal E$-Wick
rotation, which canonically associates a Euclidean $\mathcal E$-model
with every Lorentzian one. In the families studied here, this
construction induces natural analytic continuations relating the
Lorentzian and Euclidean Lax representations and renormalization-group
flows, while preserving Poisson--Lie duality.

As the principal example, we construct the Euclidean bi-Yang--Baxter
model on the Lu--Weinstein double and determine explicitly its action,
Lax representation, Poisson--Lie dual and one-loop renormalization-group
flow. We identify a one-parameter family of one-loop RG fixed points for which
the target-space geometry develops singularities on the maximal torus of
$K$, whereas the dual model with target $K^{\mathbb C}/K$ has an
everywhere regular background. This regular dual model turns  out to be a one-parameter deformation of the non-unitary 
hyperbolic Wess--Zumino--Witten model.
}

\begin{document}
\maketitle
\flushbottom
% -----------------------------------------------------------------
% Préambule modernisé (LaTeX2e / 2025)
% -----------------------------------------------------------------

 %%%%%%%%%%%%%%%%%%%%%%%%%%%%%%%%%%%%%%%%%%%%%%%%%%%%%%%%%%%%%%
 \section{Introduction}
 %%%%%%%%%%%%%%%%%%%%%%%%%%%%%%%%%%%%%%%%%%%%%%%%%%%%%%%%%%%%%%
\label{sec:intro}

The $\mathcal{E}$-models are first-order dynamical systems introduced in
\cite{KS96a,KS96b,K15} to encode the Hamiltonian dynamics of nonlinear
$\sigma$-models related by Poisson--Lie T-duality. Their scope has since
expanded considerably: besides providing a natural language for
Poisson--Lie duality, they have proved useful in the study of integrable
$\sigma$-models \cite{S17,LV21} and in formulations of double field theory
\cite{DHT19}. In the standard setting, the operator $\E$ acting on the Lie
algebra of the Drinfeld double is self-adjoint, satisfies $\E^2=\one$, and
determines a second-order $\sigma$-model with a real Lorentzian action.

\medskip

There is, however, another natural algebraic possibility: one may consider
self-adjoint operators $\E$ satisfying $\E^2=-\one$. Such operators already
appeared in \cite{K02}, where they arose in connection with non-unitary
$\sigma$-models possessing real Euclidean actions. In that work, however,
the Euclidean theory remained a secondary development, since non-unitary
models were not the principal object of investigation. In the present paper,
we study this second possibility systematically. We refer to the resulting
first-order systems as \emph{Euclidean $\E$-models}, reserving the term
\emph{Lorentzian $\E$-models} for the usual case $\E^2=\one$.

\medskip

The distinction between the two cases is physically significant. The
$\sigma$-models associated with Lorentzian $\E$-models are unitary in the
sense that their actions are real on a Lorentzian world-sheet. After a
standard Wick rotation, however, their Euclidean actions are generally
complex, since the antisymmetric tensor term acquires a factor of the
imaginary unit. Euclidean $\E$-models, by contrast, directly produce
nonlinear $\sigma$-models with real Euclidean actions. Their inverse Wick
rotations are generically complex and hence non-unitary from the Lorentzian
point of view. A basic and historically important example is the hyperbolic
Wess--Zumino--Witten model with target $K^{\bC}/K$ \cite{GK89}.

\medskip

The study of such theories has acquired new relevance through the recent
work \cite{GKR25}, which developed a powerful probabilistic framework for
the quantization of hyperbolic Wess--Zumino--Witten models by means of the
Euclidean path integral.These new methods may well prove applicable more broadly to non-unitary
$\sigma$-models with real Euclidean actions. This provides an important
motivation for the present work. In particular, one of our principal results
is the construction of a one-parameter family of conformal deformations of
the hyperbolic Wess--Zumino--Witten model which continue to possess real
Euclidean actions and retain both integrability and Poisson--Lie duality.
These theories may therefore provide useful laboratories for investigating,
within a genuinely Euclidean path-integral framework, how integrability and
duality manifest themselves at the quantum level.

\medskip

Our aim is to develop an intrinsic Euclidean counterpart of the modern
$\E$-model formalism and to determine which of its principal Lorentzian
structures survive when the relation $\E^2=\one$ is replaced by
$\E^2=-\one$. We derive the corresponding first- and second-order dynamics,
formulate Euclidean Poisson--Lie T-duality, establish sufficient conditions
for Lax integrability, and describe the one-loop renormalization of the
Euclidean $\E$-operator. Many familiar Lorentzian constructions persist, but
in modified form, with the essential differences ultimately originating in
the change of sign in the square of $\E$.

\medskip

In developing this formalism, we make the Lorentzian and Euclidean theories
resemble one another as closely as possible, so that both their analogies and
their genuine differences remain transparent. For perfect Drinfeld doubles, we show that
every Lorentzian $\E$-model admits a canonically associated Euclidean partner
through what we call the $\E$-Wick rotation. We do not, however, have a
general argument showing that this construction automatically transports
solutions of the sufficient integrability conditions or the beta functions
of strictly renormalizable families. In particular, it is not known whether
a Lax representation of an arbitrary integrable Lorentzian $\sigma$-model
necessarily induces a Lax representation of its Euclidean partner.

Nevertheless, in all the examples considered in this work, the corresponding $\E$-Wick rotated
Lorentzian and Euclidean $\sigma$-models are related by natural analytic
continuations that carry both the Lax representation and the
renormalization-group flow from one theory to the other. Such continuations
could, in principle, have been discovered directly at the level of the
$\sigma$-models, without introducing Euclidean $\E$-models. The $\E$-model
formalism, however, derives them from first principles and thereby eliminates
the guesswork involved in identifying the appropriate continuation.
Moreover, it guarantees that analytic continuations of mutually
Poisson--Lie dual Lorentzian models remain mutually Poisson--Lie dual on the
Euclidean side.

\medskip

The Euclidean bi-Yang--Baxter model provides the principal illustration of
these ideas. We construct it as the $\E$-Wick rotation of the standard
Lorentzian bi-Yang--Baxter model and determine explicitly its first-order
$\E$-operator, second-order action, Lax representation, Poisson--Lie dual
and one-loop renormalization-group flow. We study the fixed points of this flow and we find a one-parameter family of them such that 
the Poisson--Lie dual $\sigma$-model on $K^{\bC}/K$ defines a one-parameter family of
deformations of the hyperbolic Wess--Zumino--Witten model. This family is thus one-loop conformal   modulo target-space diffeomorphisms and $B$-field gauge transformations. Whether this conformal property persists beyond one
loop remains an open question.

\medskip

This example also reveals an important global feature of Euclidean
Poisson--Lie duality. For the fixed points in question, the compact-target
Euclidean bi-Yang--Baxter description becomes singular because the operator
defining its second-order background ceases to be invertible. Its
Poisson--Lie dual at the same fixed points, however, remains globally regular. It means that it 
is globally smooth and its target-space metric is nondegenerate and
definite.  Thus the same Euclidean family that is represented
by a singular $\sigma$-model in the compact frame admits a
smooth description in the non-compact dual frame. This phenomenon is
naturally visible in the Euclidean $\E$-model formalism but is not apparent
from the analytic continuation of the compact-target $\sigma$-model alone.

\medskip

The paper is organized as follows. In Section~\ref{sec:Emod}, we introduce
Euclidean $\E$-models in the first-order formalism and derive the associated
second-order nonlinear $\sigma$-models. We formulate Euclidean Poisson--Lie
T-duality and make it explicit for perfect Drinfeld doubles, where the dual
targets can be identified with mutually dual Poisson--Lie groups. In
Section~\ref{sec:EWick}, we define the $\E$-Wick rotation, construct the
Euclidean partner of a Lorentzian $\E$-model, and determine the induced
transformation of the target-space data.

Sections~\ref{sec:integr} and~\ref{sec:renorm} are devoted respectively to
integrability and renormalization. We first formulate a Euclidean analogue
of the standard sufficient condition for the integrability of Lorentzian
$\E$-models and construct the corresponding Euclidean Lax pair. We then
derive the one-loop renormalization flow of the Euclidean $\E$-operator and
compare it with its Lorentzian counterpart. Although we do not establish a
general theorem, we show that, for a particular parameter-dependent ansatz
and under suitable analyticity assumptions, the $\E$-Wick rotation respects
both Lax integrability and one-loop renormalizability: the Euclidean Lax
operators and beta functions are obtained by analytic continuation from
their Lorentzian counterparts, accompanied, where necessary, by a
reparametrization of the spectral parameter.

Finally, in Section~\ref{sec:biYB}, we apply the general formalism to the
Lu--Weinstein double. We construct and analyze the Euclidean
bi-Yang--Baxter model and its Poisson--Lie dual, establish their integrability,
and determine their renormalization-group flow. By studying the fixed points
of this flow, we obtain the one-parameter family of one-loop conformal
deformations of the hyperbolic Wess--Zumino--Witten model described above.

%%%%%%%%%%%%%%%%%%%%%%%%%%%%%%%%%%%%%%%%%%%%%%%%
\section{Lorentzian  and Euclidean \texorpdfstring{$\E$}--models}
%%%%%%%%%%%%%%%%%%%%%%%%%%%%%%%%%%%%%%%%%%%%%%%%% 
\label{sec:Emod}
\subsection{First-order formalism}
\label{sec:2.1}
 We first recall that an $\E$-model, whether Lorentzian or Euclidean, can be viewed as a first-order dynamical system, meaning that its action in some local Darboux coordinates $p_j,q_j$ on the phase space can be written as
 \be {\cal S}=\int dt \left(\sum_jp_j\dot q_j-H(p,q)\right).\label{158}\ee
 Actually, in the case of the $\E$-models, we need not resort to the local coordinates, since there exist natural global coordinates $\ell(\si)\in LD$ on the phase space, where $LD$ stands for the loop group of the Drinfeld double $D$.  Recall   that an even-dimensional Lie group $D$ is called the Drinfeld double if its  Lie algebra $\D$  is equipped with a non-degenerate, maximally indefinite, invariant symmetric bilinear form $(.,.)_\D$. The action \eqref{158} of the $\E$-model then reads\footnote{ 
 Note that the quantity $\int d\si \biggl(d \ell\ell^{-1},[ \pa_\si\ell\ell^{-1},d \ell\ell^{-1}]\biggr)_\D$ can be interpreted as a   $2$-form on the loop group $LD$, therefore the expression $d^{-1}\int d\si \biggl(d \ell\ell^{-1},[ \pa_\si\ell\ell^{-1},d \ell\ell^{-1}]\biggr)_\D$  is to be understood as a $1$-form $f(t)dt$ for some function $f(t)$ therefore the
    integral $\int d^{-1} ...$ in \eqref{120} is effectively over the time variable.   
 }
\begin{multline}
{\cal S}_{\E}(\ell)=
\jp \int dt d\si\,
(\pa_t\ell\ell^{-1},\pa_\si\ell\ell^{-1})_\D
+\js\int d^{-1}\int d\si\,
(d\ell\,\ell^{-1},[\pa_\si\ell\ell^{-1},d\ell\,\ell^{-1}])_\D
\\
-\jp \int dt d\si\,
(\pa_\si\ell\ell^{-1},\E \pa_\si\ell\ell^{-1})_\D,\label{120}
\end{multline}
where $\E:\D\to\D$ is an $\bR$-linear operator squaring either to plus or to minus identity  $\E^2=\pm\one$ and  such 
that the bilinear form $(.,\E.)_\D$ is symmetric. Here the sign plus refers to the Lorentzian 
 $\E$-model while the sign minus to the Euclidean one. 

 \medskip

 It is worth noting that for every fixed element $f\in D$, the operator
 \be \E_f:=\Ad_{f^{-1}}\E\Ad_f,\label{225}\ee
 defines essentially the same dynamical system up to a simple field redefinition  $\ell\to f\ell$. Indeed, we easily find
 \be \cS_\E(f\ell)=\cS_{\E_f}(\ell).\label{227}\ee
 In the language of the double field theory, this field redefinition is referred to as a "generalized diffeomorphism".

 \medskip

 In the Lorentzian case, it is moreover required
 that  the bilinear form $(.,\E.)_\D$ is strictly positive definite. This assumption
 guarantees the positivity of the Hamiltonian 
\be H_\E=\jp \int d\si\,
(\pa_\si\ell\ell^{-1},\E \pa_\si\ell\ell^{-1})_\D.\label{184}\ee
 However, in the Euclidean case $\E^2=-\one$, we show in the Theorem A.1 of the Appendix that  the quadratic  form $(.,\E.)_\D$ has always the split signature     $(\underbrace{1,\dots,1}_{n},
\underbrace{-1,\dots,-1}_{n})$. In particular, the Hamiltonian $H_\E$ is not bounded from below which indicates that the question of what  the quantization
of the Euclidean $\E$-model means if the time $t$ is considered to be the physical one is more subtle compared to the Lorentzian case.  

\medskip

The equations of motion of the $\E$-model  read 
\be \pa_t\ell\ell^{-1}-\E\pa_\si\ell\ell^{-1}=0,\label{198}\ee 
or, alternatively,
 \be \pa_t j-\pa_\si(\E j)+[j,\E j]=0, \quad j=\pa_\si\ell\ell^{-1}.\label{200}\ee

 In the Lorentzian case $\E^2=\one$, it is convenient to set $\pa_\pm=\pa_t\pm\pa_\si$ and $j_\pm= (\E\pm \one)j$ and rewrite the   equation of motion \eqref{198}, \eqref{200} as
  \be \pa_\pm\ell\ell^{-1}=j_\pm,\ee
  or as
  \be \pa_+j_--\pa_-j_+-[j_+,j_-]=0.\label{205}\ee

 \medskip

On the other hand, in the Euclidean case  $\E^2=-\one$, we set  $$\pa_{z}=\pa_t+\ri\pa_\si, \quad \pa_{\bz}=\pa_t-\ri\pa_\si, \quad j_z= (\E+ \ri\one )j, \quad j_\bz=(\E-\ri\one )j$$ and rewrite the  equations of motion \eqref{198}, \eqref{200} as
\be \pa_z\ell\ell^{-1}=j_z,\quad \pa_\bz \ell\ell^{-1}=j_\bz,  \ee
\be  \pa_z j_\bz-\pa_\bz j_z-[j_z,j_\bz]=0.\label{211}\ee
 
\begin{remark}\label{1} One can formally pass from the Lorentzian equations of motion to the Euclidean ones by the analytic continuation
$
\E \to -\ri \E
$
combined with the standard Wick rotation
$
\si \to -\ri\si
$
of the spatial coordinate, under which $\partial_\pm$ are sent to $\partial_z$ and $\partial_{\bz}$, respectively. Of course, this is not an allowed construction if one starts from an explicit real Lorentzian $\E$-operator, since the resulting operator $-\ri\E$ is no longer an endomorphism of the real Drinfeld double. Nevertheless, this formal replacement explains the close resemblance between the Lorentzian and Euclidean formulae and corroborates the naturalness of the Euclidean formalism. As we shall see below, analogous analytic continuations also provide useful checks in the discussion of Lax integrability and of the renormalization flow, whenever the relevant formulae depend analytically on the parameters involved.
\end{remark}
%%%%%%%%%%%%%%%%%%%%%%%%%%%%%%%%%%%%%%%%%%%%%
\subsection{Lorentzian and Euclidean Poisson--Lie T-duality}
 %%%%%%%%%%%%%%%%%%%%%%%%%%%%%%%%%%%%%%%%%%%%

Consider a maximally isotropic subgroup $G\subset D$, which means that the dimension of $G$ is one half of the dimension of $D$ and the restriction of the bilinear form $(.,.)_\D$ to the Lie subalgebra $\G$ vanishes.
It then turns out that  the first-order action \eqref{120} represents the first-order Hamiltonian dynamics of a certain nonlinear $\sigma$-model living on a target $D/G$. This can be seen as follows.

Consider an action $${\cal S}(m,h):={\cal S}_\E(mh), \quad m\in D, \ h\in G.$$ Clearly, the action ${\cal S}(m,h)$
is gauge invariant with respect to the gauge transformation $(m,h)\to (m g,g^{-1}h)$,
$g\in G$. Moreover, by fixing the gauge $h=1$, the action ${\cal S}(m,h)$ becomes just 
the original action ${\cal S}_\E(m)$. On the other hand, by using the Polyakov-Wiegmann formula, the action ${\cal S}(m,h)$ can be written also as 
\be {\cal S}(m,h)={\cal S}_\E(m)+\int dt  d\si \left ((\pa_t mm^{-1}-\E \pa_\si mm^{-1},\chi)_\D-\jp(\chi,\E\chi)_\D\right), \quad \chi=\Ad_m(\pa_\si hh^{-1}).\label{738}\ee
 It is convenient to write 
the fields $\pa_t mm^{-1}$ and $\pa_\si mm^{-1}$ as\footnote{In the case $\E^2=\one$,  the quantities $A,B$ are given unambiguously because, as we show in Theorem A.2 of the Appendix, the
  non-degeneracy of the bilinear form $(\Ad_m.,\E\Ad_m.)_\D$ restricted  to $\G$ 
  implies $\E(\Ad_m\G)\cap \Ad_m\G=\{0\}$.  In the case $\E^2=-\one$,  the
  non-degeneracy of the bilinear form $(\Ad_m.,\E\Ad_m.)_\D$ restricted  to $\G$ is not  automatic. However, if we assume that this bilinear form is non-degenerate then again it follows $\E(\Ad_m\G)\cap \Ad_m\G=\{0\}$ (see Theorem A.3 of   Appendix for the proof) and the quantities $A,B$ are also determined unambiguously.}
 $$\pa_t mm^{-1}=A_t+\E B_t,\quad \pa_\si mm^{-1}=A_\si+\E B_\si, \quad A_{t},A_{\si},B_{t},B_{\si}\in \Ad_m\G.$$
Now it is easy  to eliminate the Gaussian dependence of the action \eqref{738} on the variable $\chi$,   we  obtain
  \begin{multline} {\cal S}(m)= {\cal S}_\E(m)+\jp\int dt\oint d\si \left ( B_t-A_\si,\E(B_t-A_\si)\right)_\D=\\=\js\int d^{-1}\oint d\si(dm m^{-1},[\pa_\si m m^{-1},dm m^{-1}])_\D +   \jp\int dt\oint d\si \bigl( (B_t-A_\si,\E B_t)_\D+(A_t\mp B_\si,\E  B_\si)_\D\bigr)= \\ =\js\int d^{-1}\oint d\si(dm m^{-1},[\pa_\si mm^{-1},dm m^{-1}])_\D +\\+   \jp\int dt\oint d\si \Bigl(\pm(\E\cP_m\pa_t mm^{-1},\pa_t mm^{-1})_\D-(\E\cP_m\pa_\si mm^{-1},\pa_\si mm^{-1})_\D\Bigr)+
    \\+   \jp\int dt\oint d\si \Bigl(( \pa_t mm^{-1},\cP_m\pa_\si mm^{-1})_\D - ( \pa_\si mm^{-1},\cP_m\pa_t mm^{-1})_\D\Bigr).\label{217}\end{multline}
   Here $\cP_m$ is the projector characterized by its kernel and its image  
  \be \label{projector}  \mathrm{Ker}(\cP_m)=\Ad_m\G,  \quad
  \mathrm{Im}(\cP_m)=\E(\Ad_m\G) .\ee
  Of course, the  action \eqref{217}
obtained by  eliminating $\chi$ is gauge-invariant under transformations $m\to mg$, $g\in G$
which explains the interpretation that ${\cal{S}}(m)$ is the action of a $\si$-model living on the space of cosets $D/G$.

We also observe that the first and the third term on the right-hand-side of \eqref{217} are the integrals of real differential forms over the world sheet, namely of the WZW one and of the form $\jp(dm m^{-1} \stackrel{\wedge}{,} \cP_m dm m^{-1})$. Therefore, they can be together interpreted as antisymmetric tensor part of the  action of the nonlinear $\sigma$-model on the target $D/G$.
On the other hand, the second line corresponds to the metric part of the
$\sigma$-model; the plus  sign in the  notation $\pm$ obviously stands for the 
  $\sigma$-model living on the Lorentzian world-sheet, while the minus sign makes the world-sheet Euclidean.

  \medskip

 \bre {We have chosen on purpose the formalism in which the expressions for the Lorentzian and the Euclidean $\sigma$-models differ as little as possible, in this instance just by a minus sign in \eqref{217}. Note in this respect  that the action of the $D/G$ $\sigma$-model corresponding to the Lorentzian $\E$-model has never been written in the $\E$-model literature as in \eqref{217}. 
Instead, it is mostly written as 
\begin{multline} {\cal S}(m)=\js\int d^{-1}\int d\si(dm m^{-1},[\pa_\si m m^{-1},dm m^{-1}])_\D +\\ +\frac{1}{4} \int dt d\si\biggl(W^+_m\partial_+ mm^{-1}, \partial_- mm^{-1}\biggr)_\D-\frac{1}{4} \int dt d\si \biggl( \partial_+ mm^{-1},W^-_m \partial_- mm^{-1}\biggr)_\D.  \label{2nd} \end{multline}
   Here
 the operators $W_m^\pm:\D\to\D$ are  defined by
 \be W_m^\pm=(\one\pm\E)\cP_m.\label{255}\ee}\ere
 
The equations of motion in the second-order formalism read
    \be \pa_+(W^-_m \partial_- mm^{-1})-\pa_-(W^+_m\partial_+ mm^{-1})-[W^+_m\partial_+ mm^{-1},W^-_m \partial_- mm^{-1}]=0.\label{257}\ee
    This fact follows from \eqref{205} and from the equation
     \be j_\pm=(\E\pm\one)j=W^\pm_m \pa_\pm mm^{-1}.\label{259}\ee
    Note that \eqref{259} follows from the
    relation $\ell=mh$ which gives
    \be j= \pa_\si \ell\ell^{-1}=\pa_\si mm^{-1}+\chi=\cP_m\pa_\si mm^{-1}+\E\cP_m\pa_t mm^{-1}.\label{262}\ee
    Here we used also the first-order equation of motion which expresses $\chi$ in terms of $m$.

  \medskip

      It turns out that even in the Euclidean case the formula \eqref{217} can be rewritten similarly as in \eqref{2nd}, namely 
\begin{multline} {\cal S}(m) =\js\int d^{-1}\int  d\si(dm m^{-1},[\pa_\si mm^{-1},dm m^{-1}])_\D +\\ -\frac{\ri}{4} \int dt d\si\biggl(W^z_m\partial_z mm^{-1}, \partial_\bz mm^{-1}\biggr)_\D+\frac{\ri}{4} \int dt d\si \biggl( \partial_z mm^{-1},W^\bz_m \partial_\bz mm^{-1}\biggr)_\D, \label{247} \end{multline}
where
\begin{subequations}\label{286}
\begin{align}
    W^z_m \pa_z m m^{-1}&:=
 \cP_m \pa_t m m^{-1}+
\ri \cP_m \pa_\si m m^{-1} + \E\cP_m \pa_\si m m^{-1} -\ri\E\cP_m \pa_t m m^{-1},\\
W^{\bz}_m \pa_{\bz} m m^{-1}&:=  \cP_m \pa_t m m^{-1}
-\ri \cP_m \pa_\si m m^{-1} +\E\cP_m \pa_\si m m^{-1} +\ri  \E\cP_m \pa_t m m^{-1}. 
 \end{align}
 \end{subequations}
 Taking into account the $\bR$-linearity of the projectors $\cP_m$,
 we can rewrite those formulas more elegantly as
 \begin{subequations}\label{295}
\begin{align}
    W^z_m \pa_z m m^{-1}&=(\one -\ri\E)\cP_m\pa_z m m^{-1}
  \\
W^{\bz}_m \pa_{\bz} m m^{-1}&=(\one +\ri\E)\cP_m\pa_\bz m m^{-1}
 \end{align}
 \end{subequations}
 
 The Euclidean equations of motion in the second-order formalism read
\be  \pa_z (W^\bz_m \pa_\bz mm^{-1})-\pa_\bz (W^z_m \pa_z mm^{-1})-[W^z_m \pa_z mm^{-1},W^\bz_m \pa_\bz mm^{-1}]=0.\label{277}\ee
 
    This fact follows from \eqref{211} and from the equations
\begin{subequations}\label{279}
\begin{align}
j_z &= (\E + \ri\one)j 
= W^z_m \pa_z m m^{-1} \label{279a} \\
j_{\bz} &= (\E - \ri\one)j 
= W^{\bz}_m \pa_{\bz} m m^{-1} \label{279b}
\end{align}
\end{subequations}
Note that \eqref{279} follows from the
    relation $\ell=mh$ which gives
    \be j= \pa_\si \ell\ell^{-1}=\pa_\si mm^{-1}+\chi= -\E\cP_m\pa_t mm^{-1}+\cP_m\pa_\si mm^{-1}.\label{282}\ee

 We conclude by noting that the way to obtain the second-order 
$D/G$
$\sigma$-model from the first-order 
$\E$-model is the same, regardless of which maximally isotropic subgroup 
$G\subset D$
  is chosen. In particular, if $G$ and $\tilde G$
 are two maximally isotropic subgroups, then the 
$\sigma$-models on $D/G$ and $D/\tilde G$
 are dynamically equivalent. This statement is usually referred to as {\it Poisson–Lie T-duality}, and we observe that it has both a Lorentzian and a Euclidean version.

  %%%%%%%%%%%%%%%%%%%%%%%%%%%%%%%%%%%%%%%%%%%%%%%%%%%%%%%%
 \subsection{Perfect Poisson--Lie T-duality}
%%%%%%%%%%%%%%%%%%%%%%%%%%%%%%%%%%%%%%%%%%%%%%%%%%%%
\label{sec:Emod3}

It is instructive to work out the explicit actions \eqref{2nd} and \eqref{247} of the $\sigma$-models related by the 
Poisson--Lie duality in the case when $D$ is a perfect Drinfeld double, which means that there exist two maximally isotropic Lie subgroups $G,\ti G\subset D$ and the (Iwasawa) diffeomorphism $Iw:G\times \ti G\to D$ given by the group multiplication $$Iw(g,\ti g)= g\ti g, \quad g\in G, \ \ti g\in \tilde G.$$ 
It follows that the spaces of cosets $D/\ti G$ and $D/G$ can be identified with the mutually dual Poisson--Lie groups $G$ and $\ti G$.  

\medskip

 Let $S_\pm,A_\pm: \G\to \ti\G$ and $\ti S_\pm,\ti A_\pm:\ti \G\to \G$  be linear operators; $S_\pm,\ti S_\pm$ are supposed to be invertible and also symmetric,  which means that the bilinear forms $(S_\pm(.),.)_\D$ on $G$ and  $(\ti S_\pm(.),.)_\D$ on $\ti\G$ are symmetric. Similarly, 
  $A_\pm,\ti A_\pm$   are supposed to be antisymmetric but they need not be invertible. 
  We can then  construct linear operators  $\E_\pm:\D\to\D$  defined as
 \be \E_\pm \bpm \al \\ a\epm=\bpm \pm\ti S_\pm^{-1}\ti A_\pm&\ti S_\pm^{-1}\\  \mp S_\pm^{-1}&\pm S_\pm^{-1}A_\pm\epm\bpm \al  \\a \epm,\quad  \al\in\ti\G,\ a\in\G,\label{393}\ee
where  the block matrix notation refers to the direct sum decomposition (in the sense of the vector spaces) of the Lie algebra $\D$ of the perfect double $D$ 
\be \D=\ti\G\oplus  \G.\label{377}\ee 
We require that $S_\pm,A_\pm,\ti S_\pm,\ti A_\pm$ be such that the operators $\E_+$ and $\E_-$ define respectively the Euclidean and the Lorentzian $\E$-models,  that is
\be \E_\pm^2=\mp\one_\D,  \quad (\E_\pm x,y)_\D=(x,\E_\pm y)_\D, \quad x,y\in \D.\label{386}\ee 
This holds whenever
\be \ti S_\pm S_\pm\mp\ti A_\pm A_\pm=\one_\G, \quad \ti S_\pm A_\pm+\ti A_\pm S_\pm=0,\label{392}\ee
or, equivalently, whenever
\be S_\pm\ti S_\pm \mp A_\pm\ti A_\pm=\one_{\ti\G}, \quad S_\pm\ti A_\pm+A_\pm \ti S_\pm=0.\label{394}\ee
Note also that the equivalent conditions \eqref{392}, \eqref{394} can be both rewritten as
\be E_\pm^{-1}=\ti E_\pm,\label{322}\ee
where
\be  E_+:=S_++\ri A_+, \quad E_-:=S_--A_-, \quad \ti E_+:=\ti S_++\ri \ti A_+, \quad \ti E_-:=\ti S_- -\ti A_-.\ee
It is easy to solve the condition \eqref{322} in terms of $S_\pm,A_\pm$ 
 \be \E_\pm \bpm \al \\ a\epm=\bpm \mp A_\pm S_\pm^{-1}&S_\pm\pm A_\pm S_\pm^{-1}A_\pm\\  \mp S_\pm^{-1}&\pm S_\pm^{-1}A_\pm\epm\bpm \al  \\a \epm,\quad  \al\in\ti\G,\ a\in\G,\label{335a}\ee
 or in terms of $\ti S_\pm, \ti A_\pm$
  \be \E_\pm \bpm \al \\ a\epm=\bpm \pm\ti S_\pm^{-1}\ti A_\pm&\ti S_\pm^{-1}\\ \mp \ti S_\pm -\ti A_\pm \ti S_\pm^{-1}\ti A_\pm&\mp \ti A_\pm\ti S_\pm^{-1}\epm\bpm \al  \\a \epm,\quad  \al\in\ti\G,\ a\in\G.\label{335b}\ee
It turns out  that in the Lorentzian case all  self-adjoint operators $\E_-$ verifying $\E_-^2=\one$ and yielding positive definite the bilinear form $(.,\E_- .)_\D$ are of the form  \eqref{393} (see Theorem A.4 of the Appendix). The situation is however different in the Euclidean case where the manifold consisting of self-adjoint operators $\E_+$ verifying $\E_+^2=-\one$  is not entirely filled by the operators $\E_+$ of the form \eqref{393}. Nevertheless the operators $\E_+$ of the form \eqref{393}  form a dense open subset in that manifold. Said in other words, a generic Euclidean $\E$-model does have the $\E_+$-operator of the form \eqref{393}.

\medskip

Let us work out the action \eqref{2nd} on the group $G=D/\ti G$, in which case we can set a gauge $m\in G$. We have to find the explicit form of the projector $\cP^\pm_m$   characterized by its kernel $\Ad_m\ti\G$ and its image   $ \E_\pm(\Ad_m\ti\G)$.
For that, we can represent every element $x\in \D$ as
\be x=\Ad_m(\al+ a),\quad  a\in \G,\ \al\in\ti\G,\label{331}\ee
and then we decompose $x$
as \be x=\Ad_m\beta^\pm+\E_\pm(\Ad_m\ga^\pm).\label{333}\ee
Comparing \eqref{331} and \eqref{333}, we find
\be \al+a=\beta^\pm+\E_{\pm, m}\gamma^\pm,\label{335}\ee
where
\be \E_{\pm, m}:=\Ad_{m^{-1}}\E_\pm\Ad_m.\label{337}\ee
We note that the operators $\E_{\pm, m}$  verify the conditions \eqref{386}, therefore there exist generically  operators $S_{\pm,m},A_{\pm,m},\ti S_{\pm,m},\ti A_{\pm,m}$ such that the following holds  
 \be \E_{\pm, m}\bpm \al \\ a\epm=\bpm \pm \ti S_{\pm,m}^{-1}\ti A_{\pm,m}&\ti S_{\pm,m}^{-1}\\  \mp S_{\pm,m}^{-1}&\pm S_{\pm,m}^{-1}A_{\pm,m}\epm\bpm \al  \\a \epm,\quad  \al\in\ti\G,\ a\in\G.\label{393b}\ee
 The equation \eqref{335} can be therefore rewritten as
  \be   \bpm \al-\beta^\pm \\ a\epm=\bpm \pm \ti S_{\pm,m}^{-1}\ti A_{\pm,m}&\ti S_{\pm,m}^{-1}\\  \mp S_{\pm,m}^{-1}&\pm S_{\pm,m}^{-1}A_{\pm,m}\epm\bpm \gamma^\pm  \\0 \epm, \label{393c}\ee
  which gives, in particular
  \be \gamma^\pm=\mp S_{\pm,m} a.\ee
 We infer from it and from the identity
   $$\ti S_{\pm,m}A_{\pm,m}+\ti A_{\pm,m}S_{\pm,m}=0$$ that we have 
  \begin{multline} \cP_{\pm,m}(\pa_tmm^{-1})=\cP_{\pm,m}(\Ad_m(m^{-1}\pa_t m))=\mp\Ad_m\E_{\pm,m} S_{\pm,m}(m^{-1}\pa_t m)=\\=\Ad_m (A_{\pm,m}(m^{-1}\pa_t m))+\pa_t mm^{-1}.\label{347}\end{multline}
 It follows that
    \be  \E_\pm\cP_{\pm,m}(\pa_tmm^{-1})= \Ad_m  S_{\pm,m}(m^{-1}\pa_t m).\label{349}\ee
    
  Inserting \eqref{347} and \eqref{349} into \eqref{217}, we obtain the Lorentzian action on $D/\ti G=G$ of the form
  \be {\cal S}_-(m)=    \jp\int dt\oint d\si (m^{-1}\pa_+ m,(S_{-,m} -A_{-,m})m^{-1}\pa_- m)_\D,\label{352}\ee
  while the Euclidean action reads
    \be {\cal S}_+(m)=    \jp\int dt\oint d\si (m^{-1}\pa_z m,(S_{+,m} +\ri A_{+,m})m^{-1}\pa_\bz m)_\D.\label{353}\ee

In order to work out the action \eqref{2nd} on the dual group $\ti G=D/ G$,  it is convenient to set a gauge $\ti m\in \ti G$ and  to identify the explicit form of the projector $\ti\cP_{\pm,\ti m}$. This is done by repeating step by step the derivation for the target $G$ which essentially amounts to putting tildes on appropriate quantities like e.g. $\ti S_{\pm,\ti m}$, $\ti A_{\pm,\ti m}$. In particular, we find 
\be  \ti\cP_{\pm,\ti m}(\pa_t{\ti m}{\ti m}^{-1})=\mp \Ad_{\ti m} (\ti A_{\pm,\ti m}({\ti m}^{-1}\pa_t {\ti m}))+\pa_t {\ti m}{\ti m}^{-1},\label{359}\ee
    \be  \E_\pm\ti\cP_{\pm,\ti m}(\pa_t{\ti m}{\ti m}^{-1})= \mp\Ad_{\ti m}  \ti S_{\pm,\ti m}({\ti m}^{-1}\pa_t {\ti m}).\label{361}\ee
It follows that the Lorentzian action on $D/G=\ti G$ has the form
  \be \ti {\cal S}_-(\ti m)=    \jp\int dt\oint d\si (\ti m^{-1}\pa_+ \ti m,(\ti S_{-,\ti m} -\ti A_{-,\ti m})\ti m^{-1}\pa_- \ti m)_\D,\label{358}\ee
  while the Euclidean action reads
    \be \ti {\cal S}_+(\ti m)=   - \jp\int dt\oint d\si (\ti m^{-1}\pa_z \ti m,(\ti S_{+,\ti m} +\ri \ti A_{+,\ti m})\ti m^{-1}\pa_\bz \ti m)_\D.\label{360}\ee
    
    Of course, this is still not quite the end of the story. We have to work out the explicit dependence of the operators $S_{\pm,m}$, $A_{\pm,m}$
on $m$. To do this, we write the operator $\Ad_m:\D\to \D $ in the block form according to the decomposition \eqref{377}
\be \Ad_m^\D\bpm \al\\a\epm =\bpm a_m&0\\b_m&\Ad_m^\G\epm\bpm \al\\a\epm ,\quad  \Ad_{m^{-1}}^\D\bpm \al\\a\epm =\bpm a_m^{-1}&0\\-\Ad^\G_{m^{-1}}b_ma_m^{-1}&\Ad_{m^{-1}}^\G\epm\bpm \al\\a\epm. \label{367}\ee
Combining \eqref{393},\eqref{337},\eqref{393b}  and \eqref{367}, we find
    \be \bpm \pm \ti S_{\pm,m}^{-1}\ti A_{\pm,m}&\ti S_{\pm,m}^{-1}\\  \mp S_{\pm,m}^{-1}&\pm S_{\pm,m}^{-1}A_{\pm,m}\epm=\bpm a_m^{-1}&0\\-\Ad^\G_{m^{-1}}b_ma_m^{-1}&\Ad_{m^{-1}}^\G\epm
   \bpm \pm\ti S_\pm^{-1}\ti A_\pm&\ti S_\pm^{-1}\\  \mp S_\pm^{-1}&\pm S_\pm^{-1}A_\pm\epm
    \bpm a_m&0\\b_m&\Ad_m^\G\epm.\ee 
    It follows that 
       \be \ti S_{+,m}+\ri \ti A_{+,m}=\Ad^\G_{m^{-1}}(\ti S_+ +\ri \ti A_++\ri\Pi(m))a_m, \label{390} \ee
    \be \ti S_{-,m}- \ti A_{-,m}=\Ad^\G_{m^{-1}}(\ti S_- -\ti A_-+\Pi(m))a_m, \label{391}\ee
    where $\Pi(m)$  is the   Poisson--Lie structure on $G$ defined by 
    \be \Pi(m):=b_ma_m^{-1}.\ee

    Finally, using \eqref{322}, we find for the Lorentzian action
      \be {\cal S}_-(m)=    \jp\int dt\oint d\si (\pa_+ mm^{-1},( \ti S_--\ti A_-+\Pi(m))^{-1} \pa_- mm^{-1})_\D,\label{376}\ee
  while the Euclidean action reads
    \be {\cal S}_+(m)=    \jp\int dt\oint d\si (\pa_z mm^{-1},( \ti S_++\ri \ti A_++\ri\Pi(m))^{-1} \pa_\bz  mm^{-1})_\D.\label{378}\ee
  
  It remains  to work out the explicit dependence of the operators $\ti S_{\pm,\ti m}$, $\ti A_{\pm,\ti m}$
on $\ti m$. To do this, we write the operator $\Ad_{\ti m}:\D\to \D $ in the block form according to the decomposition \eqref{377}
\be \Ad_{\ti m}^\D\bpm \al\\a\epm =\bpm\Ad_{\ti m}^{\ti\G}&\ti b_{\ti m}\\0&\ti a_{\ti m}\epm\bpm \al\\a\epm ,\quad  \Ad_{{\ti m}^{-1}}^\D\bpm \al\\a\epm =\bpm\Ad_{\ti m^{-1}}^{\ti\G}&-\Ad_{\ti m^{-1}}^{\ti\G}\ti b_{\ti m}\ti a_{\ti m}^{-1}\\0&\ti a_{\ti m}^{-1}\epm\bpm \al\\a\epm . \label{367a}\ee
Combining \eqref{393}, \eqref{337}, \eqref{393b}  and \eqref{367a}, we find
    \be\bpm \pm \ti S_{\pm,\ti m}^{-1}\ti A_{\pm,\ti m}&\ti S_{\pm,\ti m}^{-1}\\  \mp S_{\pm,\ti m}^{-1}&\pm S_{\pm,\ti m}^{-1}A_{\pm,\ti m}\epm=\bpm\Ad_{\ti m^{-1}}^{\ti\G}&-\Ad_{\ti m^{-1}}^{\ti\G}\ti b_{\ti m}\ti a_{\ti m}^{-1}\\0&\ti a_{\ti m}^{-1}\epm
    \bpm \pm\ti S_\pm^{-1}\ti A_\pm&\ti S_\pm^{-1}\\  \mp S_\pm^{-1}&\pm S_\pm^{-1}A_\pm\epm
  \bpm\Ad_{\ti m}^{\ti\G}&\ti b_{\ti m}\\0&\ti a_{\ti m}\epm.\ee
      It follows that 
       \be S_{+,\ti m}+\ri  A_{+,\ti m}=\Ad^{\ti \G}_{{\ti m}^{-1}}( S_+ +\ri A_+-\ri\ti \Pi({\ti m}))\ti a_{\ti m}, \ee
    \be  S_{-,\ti m}- A_{-,\ti m}=\Ad^{\ti \G}_{{\ti m}^{-1}}( S_- - A_-+\ti \Pi({\ti m}))\ti a_{\ti m},  \ee
    where $\ti\Pi({\ti m})$  is the dual Poisson--Lie structure on $\ti G$ defined by \be \ti\Pi({\ti m}):=\ti b_{\ti m}\ti a_{\ti m}^{-1}.\ee
       
    Finally, using \eqref{322}, we find for the dual Lorentzian action
      \be \ti {\cal S}_-(\ti m)=    \jp\int dt\oint d\si (\pa_+ {\ti m}{\ti m}^{-1},(  S_-- A_-+\ti\Pi({\ti m}))^{-1} \pa_- {\ti m}{\ti m}^{-1})_\D,\label{397}\ee
  while the dual Euclidean action reads
    \be \ti {\cal S}_+({\ti m})=    \jp\int dt\oint d\si (\pa_z {\ti m}{\ti m}^{-1},(  -S_+-\ri  A_++\ri\ti\Pi({\ti m}))^{-1} \pa_\bz  {\ti m}{\ti m}^{-1})_\D.\label{399}\ee

    %%%%%%%%%%%%%%%%%%%%%%%%%%%%%%%%%
    \section{ \texorpdfstring{$\E$}--Wick rotation}
    %%%%%%%%%%%%%%%%%%%%%%%%%%%%%%
    \label{sec:EWick}
  Throughout, the index $+$ refers to the Euclidean case $\mathcal E_+^2=-\one$, while $-$ refers to the Lorentzian case $\mathcal E_-^2=\one$.

\medskip
  
  Consider a Lorentzian   $\E_-$-model on a Drinfeld double determined by a  self-adjoint involution 
$\E_-:\D\to\D$ such that the bilinear form $(.,\E_- .)_\D$ is strictly positive definite
(the Lorentzian case). As a consequence, the $\pm 1$-eigenspaces $V_\pm$ of $\E_-$
are $(.,.)_\D$-orthogonal and have the same dimension, equal to one half of the
dimension of the Drinfeld double.

One may ask whether one can naturally  associate a Euclidean $\E_+$-model
to a given Lorentzian $\E_-$-model. The answer is in general negative
but if the Drinfeld double is perfect such a natural association does exist. Let us describe how it works.

\medskip

The basic point of the construction of the Euclidean operator $\E_+$ out of a Lorentzian one $\E_-$ is the fact that the eigenspaces $V_\pm$ are transverse to $\ti\G$, which means that 
$\D=V_\pm\oplus \ti\G$. Indeed, non-vanishing elements $x_\pm\in V_\pm$ cannot belong to the isotropic subspace $\ti\G$ because the strict positive definiteness of the 
bilinear form $(.,\E_-.)_\D$ implies
$$(x_\pm,x_\pm)_\D=\pm (x_\pm,\E_-x_\pm)_\D\neq 0.$$
Keeping in mind the perfect decomposition 
$$\D=\ti\G\oplus\G,$$
it follows that $V_\pm$ are necessarily graphs  of some linear operators from $\ti\G\to\G$; indeed \eqref{393} gives
$$V_\pm=\{\ti x+(\ti A_-\pm \ti S_-)\ti x, \ \ti x\in\ti\G\}.$$
Note that the possibility of expressing the subspaces $V_\pm$ as graphs
gives the geometric interpretation of the algebraic results contained in Theorem A.4 of  Appendix. The mutual orthogonality of $V_+$ and $V_-$ requires that $\ti S_-:\ti\G\to\G$  and $\ti A_-:\ti\G\to\G$ present in \eqref{393} must be  symmetric  
 and  antisymmetric operators, respectively. The graph property requires  that  the operators $\ti S_-\pm \ti A_-$ be invertible
and the positive (negative) definiteness of the bilinear form $(.,.)_\D$ restricted to $V_+(V_-)$ requires 
  the bilinear form $(\ti S_- .,.)_\D$ on  $\ti \G$ must be positive definite.  

We observe that there exist natural maps $\iota_\pm: V_\pm\to V_\mp$ which are inverse to each other and which are defined by
$$\iota_\pm(\ti x+(\ti A_-\pm \ti S_-)\ti x)=\ti x+(\ti A_-\mp \ti S_-)\ti x.$$
We then define the $\E$-Wick rotated   operator $\E_+$ by
\be \E_+x_\pm=\pm \iota_\pm(x_\pm), \quad x_\pm\in V_\pm.\label{447}\ee
It is easy to check that $\E_+$ fulfils all requirements in order to define the Euclidean $\E$-model, in particular it is  self-adjoint and  it squares to $-\one$.

Let us come back to   \eqref{335b}
% \be \E_\pm \bpm \al \\ a\epm=\bpm \mp A_\pm S_\pm^{-1}&S_\pm\pm A_\pm S_\pm^{-1}A_\pm\\  \mp S_\pm^{-1}&\pm S_\pm^{-1}A_\pm\epm\bpm \al  \\a \epm,\quad  \al\in\ti\G,\ a\in\G.\label{335c}\ee
 \be \E_\pm \bpm \al \\ a\epm=\bpm \pm\ti S_\pm^{-1}\ti A_\pm&\ti S_\pm^{-1}\\ \mp \ti S_\pm -\ti A_\pm \ti S_\pm^{-1}\ti A_\pm&\mp \ti A_\pm\ti S_\pm^{-1}\epm\bpm \al  \\a \epm,\quad  \al\in\ti\G,\ a\in\G.\label{335d}\ee
 We wish to express the operators  $\ti S_+,\ti A_+$ induced by the definitions \eqref{393} and \eqref{447} in terms of the operators
  $\ti S_-,\ti A_-$. The result turns out to be very simple, the $\E$-Wick rotation  amounts to 
 %$$S_+=-S_-, \quad A_+=A_-,$$
 %or, equivalently, to
  \be \ti S_+=\ti S_-, \quad \ti A_+=-\ti A_-.\label{461}\ee
This gives an alternative way of expressing  the operator $\E_-$ as well as its $\E$-Wick rotated counterpart  $\E_+$ as 
  \be \E_\pm \bpm \al \\ a\epm =\bpm  -\ti S_-^{-1}\ti A_-& \ti S_-^{-1} \\  \mp \ti S_- -\ti A_- \ti S_-^{-1}\ti A_- &\ti A_-\ti S_-^{-1}\epm\bpm \al  \\a \epm.\label{456}\ee
  At the level of the $\sigma$-models on the target $G=D/\ti G$, 
  the Lorentzian one is given by the action \eqref{376}
    \be {\cal S}_-(m)=    \jp\int dt\oint d\si (\pa_+ mm^{-1},( \ti S_--\ti A_-+\Pi(m))^{-1} \pa_- mm^{-1})_\D,\label{466}\ee
    while the $\E$-Wick rotated Euclidean $\sigma$-model action on $G$ 
is obtained from \eqref{378} and \eqref{461}
  \be {\cal S}_+(m)=    \jp\int dt\oint d\si (\pa_z mm^{-1},( \ti S_--\ri \ti A_-+\ri\Pi(m))^{-1} \pa_\bz  mm^{-1})_\D.\label{469}\ee
  Both actions are of course real valued which shows that the 
  $\E$-Wick rotation is not the same thing as the standard Wick rotation. Indeed, the latter necessarily transforms a real action into a complex one whenever there is a nontrivial $B$-field.

  \medskip
  
  We shall illustrate the formulas \eqref{466} and \eqref{469} in Section  \ref{sec:biYB} in the context of the   Yang--Baxter deformations.
  
 \remark{The $\E$-Wick rotation can then be understood as the
analytic continuation $\E(\ti S_-, \ti A_-)
\to -\ri\E(-\ri\ti S_-, \ti A_-)$. This new operator is
still symmetric but now satisfies the Euclidean condition of squaring to minus the identity. Crucially, the structure of $\E(\ti S_-, \ti A_-)$ in \eqref{456} is such that
$-\ri\E(-\ri\ti S_-, \ti A_-)$ is still a real operator, so that the procedure maps a real
Lorentzian $\E$-model to a real Euclidean $\E$-model. In other words, this implements the formal analytic continuation $\E\to -\ri\E$ of Remark \ref{1}, together
with an additional continuation in the parameters entering $\E$ chosen such
that the resulting operator is still real.}

%%%%%%%%%%%%%%%%%%%%%%%%%%%%%%%%%%%%%%%%%%
\section{Integrability of Euclidean  \texorpdfstring{$\E$}--models}
%%%%%%%%%%%%%%%%%%%%%%%%%%%%%%%%%%%%%%%%%%
\label{sec:integr}
%%%%%%%%%%%%%%%%%%%%%%%%%%%%%%%%%%%%%%%%
\subsection{Euclidean integrability condition}
%%%%%%%%%%%%%%%%%%%%%%%%%%%%%%%%%%%%%%%%%%%

Following \cite{S17}, a Lorentzian \(\E_-\)-model is integrable if there exists a Lie algebra $\K$ and a one-parameter family of linear maps
\(O_-(\lm):\D\to\K\) such that
\begin{equation}
\label{381}
[O_-(\lm)x, O_-(\lm)\E_- x]_\K = O_-(\lm)[x,\E_- x]_\D, \qquad \forall\,x \in \D.
\end{equation}
The Lax connection is then given by
\be L_\pm(\lm)=O_-(\lm)j_\pm= O_-(\lm)\, W_m^\pm\, \pa_\pm m m^{-1}.\label{291}
\ee
This means that  for every value of the spectral parameter \(\lm\) the following  holds on shell
\[
\pa_+ L_-(\lm) - \pa_- L_+(\lm) + [L_-(\lm),L_+(\lm)]_\K = 0 .
\]

We now show that the same sufficient criterion admits a Euclidean analogue.
Namely, suppose that there exists a Lie algebra $\K$ and a one-parameter
family of linear maps \(O_+(\lm):\D\to\K\)  
 such that
\begin{equation}
\label{562}
[O_+(\lm)x, O_+(\lm)\E_+ x]_\K = O_+(\lm)[x,\E_+ x]_\D, \qquad \forall\,x \in \D.
\end{equation}
Then the Euclidean $\E_+$-model is integrable.

\medskip

The Euclidean Lax connection is expressed in terms of a solution $O_+(\lm)$ as
\be L_z(\lm)=O_+(\lm)\E_+ j+\ri O_+(\lm)j \equiv O_+(\lm)j_z , \quad  L_\bz(\lm)=O_+(\lm)\E_+ j-\ri O_+(\lm)j \equiv O_+(\lm)j_\bz. \label{512}\ee
This  means that  for every value of the spectral parameter \(\lm\) the following flatness condition holds on shell
\be 
\pa_\bz L_z(\lm) - \pa_z L_\bz(\lm) + [L_z(\lm),L_\bz(\lm)]_{\K^\bC} = 0.
\label{503}\ee

We prove \eqref{503} by combining the Euclidean equations of motion 
\eqref{200} with the condition \eqref{562}. We find 
 \begin{multline}
     \pa_\bz L_z(\lm) - \pa_z L_\bz(\lm) + [L_z(\lm),L_\bz(\lm)]_{\K^\bC} =
     \pa_\bz (O_+(\lm)\E_+ j+\ri O_+(\lm)j)\\ - \pa_z (O_+(\lm)\E_+ j-\ri O_+(\lm)j) + [O_+(\lm)\E_+ j+\ri O_+(\lm)j,O_+(\lm)\E_+ j-\ri O_+(\lm)j]_{\K^\bC} =\\=
    2\ri O_+(\lm)(\pa_t j-\pa_\si(\E_+ j)+[j,\E_+ j]). \label{508}
 \end{multline}
 We thus infer from \eqref{508} and the equations of motion \eqref{200},
 that indeed \eqref{503} holds on shell.

 \medskip

 The Lax connection in the second-order $\sigma$-model formalism is obtained by inserting  $j$ from \eqref{282} into \eqref{512}. If we take into account the $\bR$-linearity of $O_+(\lm)$, the result can be written as
  \be L_z(\lm)= O_+(\lm)W_m^z\pa_zmm^{-1}, \quad  L_\bz(\lm) =O_+(\lm)W_m^\bz\pa_\bz mm^{-1}.\label{520}\ee

  \begin{remark}  \small  The Lax connection \eqref{291} is $\K$-valued. If $\K$ is not itself a complex Lie algebra, allowing the spectral parameter to take complex values generally requires passing from $\K$ to its complexification $\K^\bC$, that is the image of the operators $O_-(\lm)$ has to be $\K^\bC$.  The Lax connection then becomes $\K^\bC$-valued and the flatness condition must be formulated with the commutator $[.,.]_{\K^\bC}$. We may nevertheless  impose a reality condition  
\be O_-(\bar\lm)x=\overline{O_-(\lm)x},  \label{642}\ee
where the  wide bar on the right-hand-side of \eqref{642} means the conjugation on $\K^\bC$ induced by the real form $\K$. In particular, if the real form $\K$ is a compact simple Lie algebra then the conjugation is (up to a sign) the standard Hermitian conjugation. 

\smallskip

The reality condition \eqref{642} guarantees that the components of the Lorentzian Lax connection verify
\be  \overline{L_\pm(\lm)}=L_\pm(\overline\lm).\label{729}\ee
In particular,  
for real $\lm$ we recover the $\K$-valuedness of the  Lax connection \eqref{291}.

\smallskip

In the Euclidean case the natural reality condition is also
given by \eqref{642}, that is 
\be O_+(\bar\lm)x=\overline{O_+(\lm)x}.  \label{803}\ee
If this condition holds the components of the Euclidean Lax connections verify  
\be \overline{L_z(\lm)}=L_{\bz}(\overline \lm).\ee
\end{remark}

%%%%%%%%%%%%%%%%%%%%%%%%%%%%%%%%%%%%%%%%%%%%%%%%
 \subsection{ \texorpdfstring{$\E$}--Wick rotation of the Lax pair}
%%%%%%%%%%%%%%%%%%%%%%%%%%%%%%%%%%%%%%%%%%%%%%%%
\label{sec:4.2}
 We may ask the following natural question: Does the integrability of a Lorentzian $\E_-$-model imply the integrability of the $\E$-Wick rotated Euclidean $\E_+$-model? 

 \smallskip

 As far as the answer to this question is concerned we have not succeeded to show in full generality that  if a Lorentzian $\E_-$-model fulfils the sufficient condition of integrability \eqref{381} for some family of operators $O_-(\lm)$ then  its $\E$-Wick rotated Euclidean counterpart automatically
satisfies the Euclidean integrability condition \eqref{562} for some related family of operators $O_+(\lm)$. However,  there are particular situations where  $O_+(\lm)$ can be constructed from the knowledge of $O_-(\lm)$. The argument is based on analyticity and
it starts with an assumption that there exists a real parameter $w$ and a one-parameter family of the integrable  Lorentzian
$\E$-models based on the
operators $\E_-(w)$ of the form (cf. \eqref{456})
  \be \E_-(w) \bpm \al \\ a\epm =\bpm  -(w\ti S_-)^{-1}\ti A_-& (w\ti S_-)^{-1} \\  + w\ti S_- -\ti A_- (w\ti S_-)^{-1}\ti A_- &\ti A_-(w\ti S_-)^{-1}\epm\bpm \al  \\a \epm.\label{689}\ee
  It is then easy to check that the one-parameter family of the $\E$-Wick rotated $\E$-models is characterized by the $\E_+(w)$ operators  of the form
  $$
\E_{+}(w)=-\ri\,\E_{-}(-\ri w).
$$
We wish to show that the Euclidean models based on $\E_+(w)$ are also integrable. 

\smallskip

We already supposed that there exists a one-parameter family of operators $O_-(\lm,w)$ verifying the integrability condition \eqref{381} for each $w$. We shall suppose also that the matrix elements of
$O_-(\lm,w)$ are analytic in $w$ in a suitable domain containing  a real interval. 
The complex analytic continuation $O_-(\lm,-\ri w)$  then satisfies 
\be [O_-(\lm,-\ri w)x,O_-(\lm,-\ri w)\E_+(w)x]_{\K^\bC}=O_-(\lm,-\ri w)[x,\E_+(w) x]_{\D^\bC}.\label{695}\ee
To show that the Euclidean  $\E_+(w)$ models  are integrable, it   is thus sufficient to choose
\be 
O_+(\lm,w):=O_-(\Lambda(\lm,w),-\ri w), \label{706}
\ee
where $\Lambda(\lm,w)$ can be an arbitrary non-constant function analytic in the domain of definition of $O_-(\lm,w)$.
Indeed, the equation 
 \eqref{695} then guarantees the on shell flatness of
 the $\K^\bC$-valued flat Lax connection
 $$ L_z(\lm)= O_+(\lm,w)j_z , \quad  L_\bz(\lm)= O_+(\lm,w)j_\bz.  $$
 
\smallskip

\begin{remark}\small 

The most obvious choice $\Lambda(\lm,w)=\lm$ may not satisfy the Euclidean reality condition \eqref{803}, and we may therefore have to choose a different parametrization of the spectral parameter. We shall see an example of this situation in Section \ref{sec:biYB}.
    
\end{remark}

%%%%%%%%%%%%%%%%%%%%%%%%%%%%%%%%%%%%%%%%%%%%%%%%%
\section{Renormalization of Euclidean  \texorpdfstring{$\E$}--models}
%%%%%%%%%%%%%%%%%%%%%%%%%%%%%%%%%%%%%%%%%%
\label{sec:renorm}
  %%%%%%%%%%%%%%%%%%%%%%%%%%%%%%%%%%%%%%%%%%%
  \subsection{Euclidean RG flow}
  %%%%%%%%%%%%%%%%%%%%%%%%%%%%%%%%%%%%%%%%%%%
  \label{sec:5.1}
 The Lorentzian $\E$-models of the form \eqref{2nd} are automatically one-loop renormalizable which means that
  the one-loop quantum corrections result in the RG flow of the 
operator $\E$. The explicit flow formula was derived in \cite{SST10} and it respects the self-adjointness  and the involutivity of the operator $\E$. It reads
\be\frac{d\E}{d\tau}=\frac{1}{4}(\E\M\E- \one\M\one), \quad \M=[[\E,\E]]-[[\one,\one]].\label{fl+}\ee
In the  Euclidean case   $\E^2=-\one$, the formula \eqref{fl+}
gets modified as   \cite{SST25}  
\be\frac{d\E}{d\tau}=\frac{1}{4}(\E\M\E+ \one\M\one), \quad \M=[[\E,\E]]+[[\one,\one]].\label{fl-}\ee
Here $[[.,.]]$ is a bilinear operation which associates to two linear operators $\cP_1,\cP_2:\D\to \D$ certain linear operator $[[\cP_1,\cP_2]]:\D\to \D$. This operation can be defined invariantly as follows: if  we write $\cP_1,\cP_2$ in the Sweedler way as
$$\cP_{j}x=\cP'_j(\cP''_j,x)_\D, \quad \cP'_j,\cP''_j\in\D, \quad j=1,2,$$ then
$$[[\cP_1,\cP_2]]x=[\cP'_1,\cP'_2]([\cP''_1,\cP''_2],x)_\D.$$
If instead we pick a  basis $T_A$ of $\D$,  then $[[\cP_1,\cP_2]]$ is the operator  defined via its matrix elements as,
 $$ [[\cP_1,\cP_2]]^{AB}=\cP_1^{KM}\cP_2^{LN}F_{KL}^{\ \ A}F_{MN}^{\ \ B}.$$
Here we use a notation
$$[T_A, T_B]=F_{AB}^{\ \ C}T_C, \quad \eta_{AB}=(T_A, \one T_B)_\D=(T_A, T_B)_\D,  \quad (\cP_j)_{AB}=(T_A,\cP_j T_B)_\D, \quad j=1,2, $$ and we define the inverse matrix $\eta^{KL}$, with the help of which we can raise the indices. In particular,
we have 
$$\cP_{j}^{AB}=\eta^{AC}\eta^{BD}(\cP_{j})_{CD}.$$

\medskip

 Note that the bracket $[[\one,\one]]$ of the identity operator with itself is a non-trivial linear operator on $\D$ given entirely by the structure of the double $\D$. 

We shall illustrate the use of the formulas \eqref{fl+} and \eqref{fl-} in Section  \ref{sec:biYB} in the context of the bi-Yang--Baxter deformations.
%%%%%%%%%%%%%%%%%%%%%%%%%%%%%%%%%%%%%%%%%%%%%%%%%%%%%
\subsection{Strict renormalizability}\label{sec:5.2}
%%%%%%%%%%%%%%%%%%%%%%%%%%%%%%%%%%%%%%%%%%%%%%%%%%%%%
All $\E$-models, Lorentzian or Euclidean, are automatically  renormalizable
since the RG flow equations \eqref{fl+}, \eqref{fl-}  have always solutions labeled by initial conditions. Moreover, if the RG flow $\E(\tau)$ is bound to some predetermined set of $\E$-operators  modulo generalized diffeomorphisms we say that the model is strictly renormalizable with respect to this predetermined set\footnote{If it is clear from the context which is the predetermined set, we say just that the model is strictly renormalizable.}. 
 
\smallskip

As an example of the strict renormalizability, let the predetermined set be a one-parameter family  $\E_s(p)$  of the  $\E$-operators 
(here $s=+$ corresponds to the Euclidean case while $s=-$ to the Lorentzian one). The parameter $p$ is a real parameter and runs over a suitable interval.

\smallskip 

The family $\E_s(p)$ is thus strictly renormalizable if there exists a real function $p_s(\tau)$ and a $\D$-valued function $\xi_s(\tau)$ such that $\E_s(\tau):=\E_s(p_s(\tau))$ satisfies the following modified RG flow equation
\be  \frac{d\E_s(\tau)}{d\tau}=\frac{1}{4}(\E_s(\tau)\M_s(\tau)\E_s(\tau)+s \one\M_s(\tau)\one)+[\ad_{\xi_s(\tau)},\E_s(\tau)].\label{784}\ee 
Here $$\M_s(\tau)=[[\E_s(\tau),\E_s(\tau)]]+s[[\one,\one]].$$ 

The  presence of the  $\xi_s$-term in \eqref{784} accounts for the expression "modulo generalized diffeomorphisms". It means that $\E_s(\tau)$ does not quite solve the RG flow equation \eqref{fl+}, however,   if we define a $D$-valued function $h_s(\tau)$  by the requirement
\be d_\tau h_s(\tau) h_s(\tau)^{-1}=\xi_s(\tau),\ee
and the function $\tilde \E_s(\tau)$ by
\be
\widetilde\E_s(\tau)
=
\Ad_{h(\tau)^{-1}}\E_s(\tau)\Ad_{h(\tau)},
\label{770}
\ee
then we readily check that $\tilde \E_s(\tau)$ does solve \eqref{fl+}
\[
\frac{d\widetilde\E_s}{d\tau}
=
\frac14
\left(
\widetilde\E_s\widetilde\M_s\widetilde\E_s
+s\widetilde\M_s
\right), \quad \tilde\M_s=[[\tilde\E_s,\tilde\E_s]]+s[[\one,\one]].
\]
Following the discussion around \eqref{225} and \eqref{227}, 
 the factors \(\Ad_{h_s(\tau)^{-1}}\) and \(\Ad_{h_s(\tau)}\) placed around the operator \(\E_s(\tau)\)  correspond merely to field redefinitions, or, at the level of the \(\sigma\)-model, to diffeomorphisms of the target.

%%%%%%%%%%%%%%%%%%%%%%%%%%%%%%%%%%%%%%%%%%%%%%%%%%%
\subsection{ \texorpdfstring{$\E$}--Wick rotation and strict  renormalization}
%%%%%%%%%%%%%%%%%%%%%%%%%%%%%%%%%%%%%%%%%%%%%%%%%%%
\label{sec:5.3}

  We may ask a natural  question: Does the strict  renormalizability of a Lorentzian $\E_-$-model imply the strict renormalizability of the $\E$-Wick rotated Euclidean $\E_+$-model?

  \smallskip

We are not able to answer this question in general, but we can do it for the case studied in Section \ref{sec:4.2}, where $p=w$ and
$\E_s(w)$ has the form \eqref{689}
 $$ \E_s(w) \bpm \al \\ a\epm =\bpm  -(w\ti S)^{-1}\ti A& (w\ti S)^{-1} \\  -sw\ti S -\ti A (w\ti S)^{-1}\ti A &\ti A(w\ti S)^{-1}\epm\bpm \al  \\a \epm, \quad s=\pm. $$
 The assumption that $\E_s(w_s(\tau))$, $s=\pm$  solves equation \eqref{784}   means in particular that the function
 $w_s(\tau)$ solves the differential equation of the form
 \be \frac{dw_s(\tau)}{d\tau}=\beta_s(w_s(\tau)).\label{795}\ee
 Here $\beta_s(w)$, $s=\pm$  are respectively  the   Euclidean and Lorentzian beta functions governing the respective flows of the coupling constant $w$.

 \smallskip

Now the $\E$-Wick relation 
$\E_+(w)=-\ri\E_-(-\ri w)$ and the equation \eqref{784} imply together the relation
\be \beta_+(w)=-\beta_-(-\ri w).\label{804}\ee
 We thus observe that establishing the fact that a Lorentzian and a Euclidean models are related by the $\E$-Wick rotation may permit to deduce the strict renormalizability of one from the  strict renormalizability of the other.
 We shall see how it works in practice in Section \ref{sec:biYB}.

%%%%%%%%%%%%%%%%%%%%%%%%%%%%%%%%%%%%%%%%%%%%%%%%%%%%%
\section{Euclidean bi-Yang--Baxter deformation}
%%%%%%%%%%%%%%%%%%%%%%%%%%%%%%%%%%%%%%%%%%%%%%%%%%%%
\label{sec:biYB}

The Yang--Baxter $\sigma$-model on the Lorentzian world-sheet  was introduced in \cite{K02} as a Poisson--Lie deformation of the principal chiral model. Its two-parameter extension, now usually referred to as the bi-Yang--Baxter or Klim\v c\'\i k $\sigma$-model, was proposed in \cite{K09}, while the existence of its Lax pair and hence its classical integrability were established in \cite{K14}. A detailed subsequent analysis of the model, including its Hamiltonian structure, target-space geometry and one-loop renormalization, was carried out by Bazhanov, Kotousov and Lukyanov in \cite{BKL18}.

\smallskip 

An intriguing feature of the formulation used in \cite{BKL18} is the rôle played by analytic continuation of the two deformation parameters. For purely imaginary values of these parameters, with their ratio kept real, both the target-space metric and the torsion are real, and one obtains a real Lorentzian $\sigma$-model. For real values of the same parameters, on the other hand, the metric remains real whereas the torsion becomes purely imaginary, so that the corresponding Lorentzian action is complex. From the perspective developed in the present paper, this latter branch admits a natural reinterpretation: after the standard Wick rotation of the world-sheet, the extra factor of $\ri$ multiplying the antisymmetric-tensor contribution converts it into a real Euclidean action. Thus, although the Euclidean interpretation was not made explicit in \cite{BKL18}, the analytically continued model considered there already contains, in an implicit form, the Euclidean bi-Yang--Baxter $\sigma$-model and the renormalization-group flow relevant to it.

\smallskip

Related analytic continuations have appeared in the study of Yang--Baxter deformations of symmetric-space $\sigma$-models. In particular, Fateev and Litvinov \cite{FL18} and Litvinov and Spodyneiko \cite{LS18} investigated dual descriptions of the $\eta$-deformed $O(N)$ models, while Bychkov and Litvinov \cite{BL25} recently analyzed a new ``dual regime'' of Yang--Baxter-deformed $O(2N)$ models. These theories are coset models rather than $\sigma$-models on group manifolds and would naturally be described, in the $\E$-model language, by dressing cosets. Their analyses involve expansions around the complex fixed point $\eta=\ri$, which, from the viewpoint advocated here, is naturally interpreted as a Euclidean fixed  point.

\smallskip 

The purpose of the present work is not merely to rephrase these analytic continuations. Instead of beginning with a Lorentzian action, allowing its couplings to become complex and subsequently searching for a real Euclidean slice, we develop from the outset a real Euclidean $\E$-model formalism. Within this framework, the Euclidean Yang--Baxter and bi-Yang--Baxter $\sigma$-models arise directly from real first-order data. Their Poisson--Lie duality, classical integrability and one-loop renormalization can consequently be studied without passing through a complex Lorentzian theory. This construction places the previously observed analytically continued regimes into a systematic geometric and algebraic framework.

\smallskip

The $\E$-model perspective also reveals phenomena which are not visible from the direct analytic continuation of the conventional $\sigma$-model action. For the Euclidean Yang--Baxter model, the second-order target-space geometry becomes singular at the Euclidean fixed point, however,  its Poisson--Lie T-dual (as well as the underlying Euclidean $\E$-model) remains regular there. The resulting dual theory is a deformation of the hyperbolic Wess--Zumino--Witten model on $\K^\bC/\K$.

\smallskip

An analogous, and even more remarkable, phenomenon occurs for the Euclidean bi-Yang--Baxter model. Its target-space geometry may again become singular at the Euclidean fixed points, although the corresponding first-order $\E$-operator remains perfectly regular. Passing to the Poisson--Lie dual description removes this apparent geometric degeneration and yields a smooth deformation of the hyperbolic Wess--Zumino--Witten background. This deformation is not only classically integrable and Poisson--Lie dualizable: it remains one-loop conformal for every value of the deformation parameter. The identification of this non-trivial one-loop conformal family constitutes one of the principal new results of the present paper.

\smallskip 

This result may also be relevant to the quantum theory. Guillarmou, Kupiainen and Rhodes have recently given a rigorous probabilistic construction of the path integral of the $\mathbb H^3$ Wess--Zumino--Witten model \cite{GKR25}. The smooth one-loop conformal deformation obtained here therefore provides a natural candidate for investigating whether their probabilistic methods can be extended beyond the undeformed hyperbolic WZW point. We do not address this difficult question in the present work, but the Euclidean $\E$-model construction supplies a distinguished family of real, one-loop conformal and classically integrable actions on which such an extension may be explored.

\smallskip

The plan of the present Section \ref{sec:biYB} is as follows: First we review the concept of the Lu-Weinstein Drinfeld double and then we study the Euclidean Yang-Baxter $\si$-model and its Poisson-Lie dual. In particular, we show that the Euclidean RG fixed  point corresponds to a singular geometry on the compact group target but the dual geometry is perfectly  smooth  and it coincides with the hyperbolic WZW model. Then we describe the Euclidean bi-Yang-Baxter deformation and  we establish its integrability and  renormalizability by using the $\E$-model formalism. Most importantly we show that the Poisson-Lie dual at the Euclidean fixed point constitutes a one-parameter one-loop conformal deformation of the hyperbolic WZW model.

%%%%%%%%%%%%%%%%%%%%%%%%%%%%%%%%%%%%%%%%%%%%%%%
\subsection{Lu--Weinstein Drinfeld double}\label{sec:6.1}
%%%%%%%%%%%%%%%%%%%%%%%%%%%%%%%%%%%%%%%%%%%%%%%

Lu--Weinstein Drinfeld double \cite{LW} is the complexification $K^\bC$
of the simple compact Lie group $K$. It is to be  viewed as a  real group, for example, $D=SU(2)^\bC=SL(2,\bC)$ is a six-dimensional real group. The non-degenerate symmetric invariant bilinear form $(.,.)_\D$ is given by the expression 
\be (x,y)_\D:=\frac{1}{\eta}\Im\tr(xy)\equiv \frac{1}{2\ri\eta}\tr(xy-x^\dagger y^\dagger),\quad x,y\in\D, \quad \eta>0, \label{bif}\ee
where the symbol $\dagger$ stands for the standard Hermitian conjugation on $\K^\bC$ and $\eta$ is a  parameter.

\smallskip 

It turns out that the Lu--Weinstein double is perfect in the sense of Section  \ref{sec:Emod3}. One of the maximally isotropic subgroups is   the compact subgroup $K\subset K^\bC$, while the other $\ti K$ is the subgroup $AN\subset K^\bC$ which appears in the
Iwasawa decomposition $K^\bC=K\ti K\equiv KAN$. 
In particular, if $K^\bC=SL(n,\bC)$ then $AN$ is formed by the upper-triangular matrices having real positive numbers on the diagonal. In the general case, it is given by the exponentiation of the Lie subalgebra $\ti\K$ 
\be \tilde\K:=(R-\ri)\K\subset \K^\bC,\label{554}\ee
where the $\bR$-linear operator $R:\K^\bC\to\K^\bC$   is the Yang--Baxter operator. It is defined as 
  $$RH^\mu=R\ri H^\mu=0,\quad RE^\alpha=- {\rm sign}{(\alpha)}\ri E^\alpha,$$
  where $H^\mu$ stand for the Cartan generators and $E^\alpha$
  are the step generators of $\K^\bC$.

  \medskip

  It is not difficult to verify that the vector subspace 
  $(R-\ri)\K$ is indeed a maximally isotropic Lie subalgebra of $\K^\bC$. For that, it is sufficient to use the antisymmetry  of the Yang--Baxter operator  
  $$\tr((Rx)y+x(Ry))=0, \quad x,y\in \K^\bC$$
 as well as the  so-called Yang--Baxter identity
  $$[Rx,Ry]=R([Rx,y]+[x,Ry])+[x,y], \quad x,y\in \K^\bC. $$
 We also infer from \eqref{554}, that there is a natural invertible linear map $\ups:\K\to\ti\K$ defined as 
\be \ups(a)=\eta(R-\ri)a, \quad a\in\K.\label{720}\ee
 For the  Lu--Weinstein double it then holds
 \be (\ups(a),b)_\D=-\tr(ab),\quad a,b\in \K.\label{574}\ee

We easily infer from \eqref{466},\eqref{469} and \eqref{574}, that  for the specific case of the Lu--Weinstein double, the general formulas
for the $\sigma$-models derived from a Lorentzian $\E$-model and  from
its $\E$-Wick rotated Euclidean counterpart become respectively 
    \be {\cal S}_-(m)=    -\jp\int dt\oint d\si \tr (\pa_+ mm^{-1}\ups^{-1}( \ti S_--\ti A_-+\Pi(m))^{-1} \pa_- mm^{-1}),\label{583}\ee
    while the $\E$-Wick rotated Euclidean $\sigma$-model action on $K$ 
is obtained from \eqref{378} and \eqref{461}
  \be {\cal S}_+(m)=    -\jp\int dt\oint d\si \tr(\pa_z mm^{-1}\ups^{-1}( \ti S_--\ri \ti A_-+\ri\Pi(m))^{-1} \pa_\bz  mm^{-1}).\label{586}\ee
 We easily find
  $$\Ad_m \ups a=\eta\Ad_m(R-\ri\one)a=\eta(\Ad_mR-R\Ad_m )a+\eta(R-\ri\one)\Ad_m a=\eta(\Ad_mR-R\Ad_m)a+\ups\Ad_m a.$$
   Recalling the definitions \eqref{367} of the operators $a_m,b_m$, 
  we thus  observe that
  $$a_m\ups a=\ups \Ad_m a, \quad b_m\ups a=\eta(\Ad_mR-R\Ad_m) a.$$
  It follows
    \be \Pi(m)\ups a =b_ma_m^{-1}\ups a=b_m \ups \Ad_{m^{-1}} a =\eta (R_{m^{-1}}-R)a, \quad
\label{588}\ee
  where
  \be  R_{m^{-1}}:=\Ad_mR\Ad_{m^{-1}}.\label{595}\ee
 Using \eqref{588} 
 we may rewrite the actions \eqref{583} and \eqref{586} as
 \be {\cal S}_-(m)=    -\jp\int dt\oint d\si \tr (\pa_+ mm^{-1}( U-V+\eta R_{m^{-1}}-\eta R)^{-1} \pa_- mm^{-1}),\label{590}\ee
\be {\cal S}_+(m)=    -\jp\int dt\oint d\si \tr (\pa_z mm^{-1}( U-\ri V+\ri\eta  R_{m^{-1}}- \ri \eta R)^{-1} \pa_\bz mm^{-1}),\label{592}\ee
 where
 \be U:=\ti S_-\ups, \quad V:=\ti A_-\ups,\label{597}\ee
Note that in terms of $U,V$, the pair of $\E$-Wick rotated $\E_\pm$ operators  \eqref{456}    can be rewritten as 
 \be \E_\pm a=VU^{-1}a +\ups U^{-1}a, \quad  \E_\pm \ups b=(\mp U-VU^{-1}V)b -\ups U^{-1}Vb, \label{612}\ee

%%%%%%%%%%%%%%%%%%%%%%%%%%%%%%%%%%%%%%%%%%%%%%%%%%%
 \subsection{Yang-Baxter deformation and hyperbolic WZW model}
 \label{sec:6.2}
%%%%%%%%%%%%%%%%%%%%%%%%%%%%%%%%%%%%%%%%%%%%%%%%%%%
 The Yang--Baxter $\E$-model  on the Lu-Weinstein double is characterized by
 the following choice of the $\E$-Wick rotated $\E_\pm$ operators
\be \E_s(a+\ri b)= \frac{sw}{\eta}b-\ri \frac{\eta}{w} a, \quad a,b\in\K, \quad s=\pm.\label{611}\ee
Here the choice  $s=+$ ($-$) corresponds to the Euclidean (Lorentzian) model.

\smallskip

Comparing with \eqref{612}, this choice corresponds to 
\be V=-\eta  R, \quad U=w\one_\K,\label{602}\ee
therefore from \eqref{590} and \eqref{592}, we obtain
 \be {\cal S}_-(m)=    -\jp\int dt\oint d\si \tr (m^{-1}\pa_+ m(w \one+ \eta R)^{-1} m^{-1}\pa_- m),\label{604}\ee
\be {\cal S}_+(m)=    -\jp\int dt\oint d\si \tr (m^{-1}\pa_z m(w \one + \eta\ri R)^{-1} m^{-1}\pa_\bz m).\label{605}\ee
We recognize in \eqref{604} the action of the standard Lorentzian Yang--Baxter $\sigma$-model \cite{K02,K09}. For that reason, we refer to the $\E$-Wick rotation \eqref{605} of \eqref{604} as the  {\it Euclidean Yang--Baxter $\sigma$-model}.

 \bre \small Note that the Euclidean action \eqref{605} does not represent the Wick rotation of the Lorentzian action \eqref{604} because of the factor $\ri$ which multiplies the Yang--Baxter operator $R$. At the same time, it is this factor $\ri$ which makes the Euclidean action real.
 \ere

In order to find the actions of the dual Yang-Baxter models,
we have to determine the precise form of the operators $S_s,A_s$, $s=\pm$ to be inserted into the general expressions \eqref{397} and \eqref{399}. 

\smallskip

Let $T^i$ be a basis of $\K$ verifying
\be \tr(T^i T^j)=-\delta^{ij}\label{1080a}\ee
and $\tau_i=\Upsilon(T^i)$ the corresponding basis of $\ti\K$ so that
\be (\tau_i,T^j)_\D=\delta^j_i.\label{1082a}\ee
Then we find from \eqref{611} 
$$ \E_sT^i=\frac{1}{w}\tau_i-\frac{\eta}{w}RT^i,\quad \E_s\tau_i= \frac{\eta}{w}R\tau_i-sw T^i-\frac{\eta^2}{w}R^2T^i$$
and from the general formula \eqref{335a}
we infer
\be S_sT^i=\frac{1}{w}\frac{w^2-s\eta^2(R^2+1)}{w^2-s\eta^2}\tau_i,\quad  A_sT^i=\frac{\eta}{\eta^2-sw^2}R\tau_i.\label{788} \ee

\smallskip

Inserting \eqref{788} into \eqref{397} and \eqref{399} and introducing a variable $P=\ti m\ti m^\dagger$, we obtain  
$$\ti {\cal S}_\pm(\ti m)=\tilde S_\pm(P),$$
where
\begin{multline}\tilde S_-(P)=-\frac{\ri}{8\eta}\oint dt d\si \tr \left(\frac{w-\ri\eta+(w+\ri\eta)\Ad_P}{w-\ri\eta-(w+\ri\eta)\Ad_P} (P^{-1}\xp P)(P^{-1}\xm P)\right) \\-\frac{\ri}{8\eta}\int d\si d^{-1}\tr(P^{-1}dP,[P^{-1}P' ,P^{-1}dP ]).\label{lyb}\end{multline}\begin{multline}\tilde S_+(P)=-\frac{1}{8\eta}\oint dt d\si \tr \left(\frac{w+\eta+(w-\eta)\Ad_P}{w+\eta-(w-\eta)\Ad_P} (P^{-1}\xz P)(P^{-1}\xbz P)\right) \\-\frac{\ri}{8\eta}\int d\si d^{-1}\tr(P^{-1}dP,[P^{-1}P' ,P^{-1}dP ]).\label{eyb}\end{multline}\bigskip
We note that the Lorentzian formula \eqref{lyb} was first obtained in \cite{K16} while the Euclidean one \eqref{eyb} is new.

\smallskip

The action \eqref{eyb} is nothing but a one-parameter deformation of the hyperbolic WZW model, where the deformation parameter is equal to
$$\frac{w-\eta}{w+\eta}$$
and the hyperbolic WZW model corresponds to the case where this parameter vanishes, i.e. $w=\eta$.

\smallskip

We know that it holds
$$\E_+(w)=-\ri \E_-(-\ri w),$$
therefore, following the observations made in Remark \ref{1}, we expect that the analytic continuations $w\to -\ri w $ and $\si\to -\ri\sigma$ of the action \eqref{604} should give the action \eqref{605}. We check easily that this is indeed true. Similarly, the same analytic continuations transform \eqref{lyb} in \eqref{eyb}. We stress, however, that had 
we not had our Euclidean $\E$-model formalism, those  analytical continuations could suggest the duality between the Euclidean actions \eqref{605} and \eqref{eyb} but it would not prove it. 

\smallskip

As far as the integrability of the $\E$-Wick rotated Yang-Baxter model is concerned, we shall establish it later on in Section \ref{sec:6.4} as the special case of the proof of the integrability of the more general bi-Yang-Baxter deformation. However, as far as the renormalizability is concerned, we prefer to give to it a special treatment right here and not to wait until the study of the renormalization of the bi-Yang-Baxter case. The reason is an interesting  property of the Euclidean Yang-Baxter RG fixed point which we now describe.

\smallskip

Let us first  show that the family $\E_s(w)$ of the Yang-Baxter operators \eqref{611} is strictly renormalizable.
Set
\[
V^a:= \ri\eta T^a ,
\]
where  \(T^a\) is a normalized basis of \(\K\), namely
\[
\tr(T^aT^b)\equiv (T^a,T^b)_\K=-\delta^{ab}.
\]
We  have
\[
(V^a,V^b)_\D=0,
\qquad
(V^a,T^b)_\D=-\delta^{ab}, \qquad (T^a,T^b)_\D=0.
\]
The  Yang--Baxter operator $\E_s(w)$ can then be rewritten as
\be
\E_s(w)(u+\ri v)
=
\frac1w V^a(V^a,u+\ri v)_\D
-sw T^a(T^a,u+\ri v)_\D ,
\ee
or, in the bra-ket formalism, as
\[
\E_s(w)=\frac1w |V^a)(V^a|-sw |T^a)(T^a|.
\]
Similarly, the identity operator \(\one\) can be written as
\be
\one=-|V^a)(T^a|-|T^a)(V^a|.
\ee
Using the identity 
\[
(T^a,[T^c,T^d])_\K(T^b,[T^c,T^d])_\K=c_\K\delta^{ab},\]
where $c_\K$ is the double of the dual Coxeter number, it is easy   to calculate\begin{multline}
[[\E_s,\E_s]]
=
\frac1{w^2}
|[V^a,V^b])([V^a,V^b]|
-s|[V^a,T^b])([V^a,T^b]|
-s|[T^a,V^b])([T^a,V^b]|
+\\+w^2|[T^a,T^b])([T^a,T^b]|=
c_{\K}\left[
\left(w^2+\frac{\eta^4}{w^2}\right)|T^a)(T^a|
-2s|V^a)(V^a|\right],
\label{EEbracket}
\end{multline}
\begin{multline}
[[\one,\one]]
=
|[V^a,V^b])([T^a,T^b]|
+|[V^a,T^b])([T^a,V^b]|
+|[T^a,V^b])([V^a,T^b]|+\\
+|[T^a,T^b])([V^a,V^b]|=
2c_{\K}
\left(
|V^a)(V^a|-\eta^2|T^a)(T^a|
\right).
\label{11bracket}
\end{multline}
Therefore 
\be\M_s=[[\E_s,\E_s]]+s[[\one,\one]]=
\frac{c_{\K}}{w^2}
\left( \eta^2-sw^2\right)^2
|T^a)(T^a|.
\label{Msfinal}
\ee
Furthermore, we have
$$
\E_s T^a=-\frac1w V^a,
\qquad
\E_s V^a=sw T^a,
$$
hence
\be \frac14\left(\E_s\M_s\E_s+s\M_s\right)
=
\frac{c_{\K}}{4w^2}
\left(\eta^2-sw^2\right)^2
\left(
\frac1{w^2}|V^a)(V^a|
+s|T^a)(T^a|
\right).\label{1181}
\ee
On the other hand, we find
\be 
\frac{\partial\E_s}{\partial w}
=
-\frac1{w^2}|V^a)(V^a|
-s|T^a)(T^a|,\label{1195}
\ee
Therefore if  a function  $w_s(\tau)$ solves the differential equation
\be
\frac{dw_s(\tau)}{d\tau}=\beta_s(w_s(\tau))
\equiv -\frac{c_{\K}}{4w_s(\tau)^2}
\left(\eta^2-sw_s(\tau)^2\right)^2 ,
\label{wflow}
\ee
then the flow $\E_s(\tau):=\E_s(w_s(\tau))$  solves 
 the RG flow equation
\be \frac{d\E_s(\tau)}{d\tau}=\frac{1}{4}(\E_s(\tau)\M_s(\tau)\E_s(\tau)+s \one\M_s(\tau)\one), \quad \M_s(\tau)=[[\E_s(\tau),\E_s(\tau)]]+s[[\one,\one]].\label{974}\ee
We conclude that the  family $\E_s(w)$ of the Yang-Baxter operators \eqref{611} is strictly renormalizable with the beta function given by
\eqref{wflow}. In particular, the formula \eqref{wflow} confirms the general result obtained in Section \ref{sec:5.3}  
 \be\beta_+(w)=-\beta_-(-\ri w).\label{1077}\ee

It is interesting to observe that 
whereas no real finite fixed point exists in the Lorentzian case, there exist fixed points in the Euclidean one. What is even more interesting is that while at the Euclidean fixed points the target space geometry of the Euclidean Yang-Baxter $\sigma$-model gets singular
(because the operators $\one\pm \ri R$ are not invertible), the Euclidean Poisson-Lie dual action \eqref{eyb}  is perfectly regular there. In particular, at the fixed point $w=\eta$, the dual action
\eqref{eyb} is that of the standard hyperbolic WZW model.

We may thus conclude that  the one-loop conformal singularity of the
Euclidean Yang-Baxter model gets resolved by passing to the Poisson-Lie dual picture.

%%%%%%%%%%%%%%%%%%%%%%%%%%%%%%%%%%%%%%%%%%%%%%%%%%%%%%%%
 \subsection{Poisson-Lie T-duality in the bi-Yang-Baxter context}
%%%%%%%%%%%%%%%%%%%%%%%%%%%%%%%%%%%%%%%%%%%%%%%%%%%%
  The bi-Yang--Baxter $\E$-model on the Lu-Weinstein double is characterized by
 the following choice of the $\E$-Wick rotated $\E_\pm$ operators
\be \E_s(w)(u+\ri v)=  -\frac{\mu}{w} Ru+\frac{w}{\eta}\left(s\one+\frac{\mu^2}{w^2}R^2\right)v-\ri \frac{\eta}{w} u+\ri\frac{\mu}{w} Rv, \quad u,v\in\K,\quad s=\pm.\label{1082}\ee
In particular, in the notation inroduced in Section \ref{sec:6.1} this choice corresponds to 
\be V=-(\eta+\mu)  R, \quad U=w\one_\K,\label{1084}\ee
therefore from \eqref{590} and \eqref{592}, we obtain
 \be {\cal S}_-(m)=    -\jp\int dt\oint d\si \tr (\pa_+ mm^{-1}( w\one +\eta R_{m^{-1}}+\mu R)^{-1} \pa_- mm^{-1}),\label{1086}\ee
\be {\cal S}_+(m)=    -\jp\int dt\oint d\si \tr (\pa_z mm^{-1}(w\one +\ri\eta R_{m^{-1}}+\ri\mu R)^{-1} \pa_\bz mm^{-1}).\label{1088}\ee
We recognize in \eqref{1086} the action of the standard Lorentzian bi-Yang--Baxter $\sigma$-model \cite{K02,K14}. For that reason, we refer to the $\E$-Wick rotation \eqref{1088} of \eqref{1086} as to the  {\it Euclidean bi-Yang--Baxter $\sigma$-model}. 

\smallskip

As far as the dual Yang-Baxter models are  concerned,
we have  to determine the precise form of the operators $S_s,A_s$, $s=\pm$ to be inserted into the general expressions \eqref{397} and \eqref{399}.  

\smallskip 

 Consider again the basis $T^i,\tau_i$ defined by the relations \eqref{1080a} and \eqref{1082a}.
We then find from \eqref{1082} 
\be \E_sT^i=\frac{1}{w}\tau_i-\frac{\eta+\mu}{w}RT^i,\quad \E_s\tau_i= \frac{\eta+\mu}{w}R\tau_i-sw T^i-\frac{(\eta+\mu)^2}{w}R^2T^i\ee
and from the general formula \eqref{335a}
 $$ \E_\pm \bpm \al \\ a\epm=\bpm \mp A_\pm S_\pm^{-1}&S_\pm\pm A_\pm S_\pm^{-1}A_\pm\\  \mp S_\pm^{-1}&\pm S_\pm^{-1}A_\pm\epm\bpm \al  \\a \epm,\quad  \al\in\ti\G,\ a\in\G, $$
we infer
\be S_sT^i=\frac{1}{w}\frac{(\eta+\mu)^2(R^2+1)-sw^2}{(\eta+\mu)^2-sw^2}\tau_i,\quad  A_sT^i=\frac{\eta+\mu}{(\eta+\mu)^2-sw^2}R\tau_i.\label{1350} \ee
Inserting \eqref{1350} into \eqref{397} and \eqref{399} and introducing a variable $P=\ti m\ti m^\dagger$, we obtain  
\be \ti {\cal S}_\pm(\ti m)=\tilde S_\pm(P),\label{1247}\ee
where 
\begin{multline}\tilde S_-(P)=-\frac{\ri}{8\eta}\oint dt d\si \tr \left(\frac{w-\mu R-\ri\eta+(w-\mu R+\ri\eta)\Ad_P}{w-\mu R-\ri\eta-(w-\mu R+\ri\eta)\Ad_P} (P^{-1}\xp P)(P^{-1}\xm P)\right) \\-\frac{\ri}{8\eta}\int d\si d^{-1}\tr(P^{-1}dP,[P^{-1}P' ,P^{-1}dP ]),\label{eyc}\end{multline}
\begin{multline}
\tilde S_+(P)
=
-\frac{1}{8\eta}\oint dt\,d\sigma\;
\tr\left[
\left(
\frac{
w-\ri\mu R+\eta
+
(w-\ri\mu R-\eta)\Ad_P
}{
w-\ri\mu R+\eta
-
(w-\ri\mu R-\eta)\Ad_P
}
\right)
(P^{-1}\partial_z P)(P^{-1}\partial_{\bar z}P)
\right]
\\
-\frac{\ri}{8\eta}
\int d\sigma\,d^{-1}
\tr\left(
P^{-1}dP,
[
P^{-1}P',
P^{-1}dP
]
\right).
\label{eyc+}
\end{multline}
Here we used  the convention
\[
\frac{A}{B}:=B^{-1}A.
\]
 Note that in the Lorentzian case the equality \eqref{1247} was proved in \cite{K15} where
 the Poisson-Lie dual of the bi-Yang-Baxter model was first written in the  form \eqref{eyc}.
 In the Euclidean case, the formula \eqref{eyc+} is a new result. Note that for $\mu=0$ and $w=\eta$ the Euclidean model \eqref{eyc+} reduces to the hyperbolic WZW model.
 %%%%%%%%%%%%%%%%%%%%%%%%%%%%%%%%%%%%%%%%%%%%%%%%%%%%%%%%
 \subsection{Integrability of the bi-Yang-Baxter   \texorpdfstring{$\E$}--model}
%%%%%%%%%%%%%%%%%%%%%%%%%%%%%%%%%%%%%%%%%%%%%%%%%%%%
\label{sec:6.4}
 Our next goal will be to find families of operators $O_s(\lm):\D\to\K$ which verify the sufficient conditions of integrability
\begin{equation}
\label{1154}
[O_s(\lm)x, O_s(\lm)\E_s x]_\K = O_s(\lm)[x,\E_s x]_\D, \quad \forall\,x \in \D, \quad s=\pm.
\end{equation}
For that it will be more convenient to parametrize every element $x$ of $\D$ as
$$x=y+\frac{1}{w}(\mu R+\ri\eta)z, \quad y,z\in\K.$$
Then we find
$$\E_sx=sz-\frac{1}{w}(\mu R+\ri\eta)y, \quad s=\pm.$$
Using the modified classical Yang–Baxter equation
\[
[Rx,Ry]=R([Rx,y]+[x,Ry])+[x,y]\equiv R[x,y]_R+[x,y],
\] we find also
\begin{multline}
[x,\E_s x]=\left(s-\frac{\eta^2}{w^2}+\frac{\mu^2}{w^2}\right)[y,z]+\frac{\mu}{w}([Ry,y]+s[Rz,z])+\frac{1}{w^2}(\mu R+\ri\eta)\mu[y,z]_R.
\end{multline}
 We look for the families of operators $O_s(\lm):\D\to\K$ in the form
 \be O_s(\lm)x=\frac{1}{w}f_s(\lm)y+\frac{1}{w}(g_s(\lm)+\mu R)z,  \quad s=\pm.\ee
 It follows
 \begin{multline} [O_s(\lm)x,O_s(\lm)\E_s x]=[\frac{1}{w}f_s(\lm)y+\frac{1}{w}(g_s(\lm)+\mu R)z,\frac{1}{w}sf_s(\lm)z-\frac{1}{w}(g_s(\lm)+\mu R)y]=\\=\left(s\frac{f_s(\lm)^2}{w^2}+\frac{g_s(\lm)^2}{w^2}+\frac{\mu^2}{w^2}\right)[y,z]+\frac{\mu}{w^2} f_s(\lm)(s[Rz,z]+[Ry,y])+\frac{1}{w^2}(\mu R+g_s(\lm))\mu[y,z]_R.\label{649}\end{multline}
 \begin{multline} O_s(\lm)[x,\E_s x]=  \frac{1}{w}f_s(\lm)\left(s-\frac{\eta^2}{w^2}+\frac{\mu^2}{w^2}\right)[y,z]+\frac{\mu}{w^2} f_s(\lm)(s[Rz,z]+[Ry,y])+\frac{1}{w^2}(\mu R+g_s(\lm))\mu[y,z]_R.\label{659}\end{multline}
The integrability condition \eqref{1154} is then fulfilled if it holds
\be s f_s(\lm)^2 + g_s(\lm)^2 + \mu^2 =\frac{f_s(\lm)}{w}\left(sw^2- \eta^2 + \mu^2 \right), \quad s=\pm.\label{656}\ee
  This condition admits rational solution  
\be
f_s(\lambda)=\frac{w^2 - s\eta^2 + s\mu^2}{2w}
+w\rho_s\,\frac{1 - s\lambda^2}{1 + s\lambda^2},
\qquad
g_s(\lambda)=\frac{2w\,\rho_s\,\lambda}{1 + s\lambda^2},
\qquad
\rho_s^2=\frac{\left(s - \frac{\eta^2}{w^2} + \frac{\mu^2}{w^2}\right)^2}{4}
- s\frac{\mu^2}{w^2}.
\label{696}
\ee
In the Lorentzian case, this gives the Lax pair \eqref{291} verifying the reality condition \eqref{642}
\be L_\pm(\lm)=O_-(\lm)j_\pm= O_-(\lm)\, W_m^\pm\, \pa_\pm m m^{-1}.\label{669}\ee
\begin{subequations}\label{688}
    \begin{align}
        L_+(\lm)=(f_-(\lm)-g_-(\lm)-\mu R)(w-\eta R_{m^{-1}}-\mu R)^{-1}\pa_+mm^{-1},\\
       L_-(\lm)=( f_-(\lm)+g_-(\lm)+\mu R)(w+\eta R_{m^{-1}}+\mu R)^{-1}\pa_-mm^{-1}.
    \end{align}
\end{subequations}

 In the Euclidean case, the Lax pair \eqref{520} verifying the reality condition \eqref{729} is
$$ L_z(\lm)= O_+(\lm)W_m^z\pa_zmm^{-1}, \quad  L_\bz(\lm) =O_+(\lm)W_m^\bz\pa_\bz mm^{-1},$$ 
where the quantities $W_m^z\pa_zmm^{-1}$, $W_m^\bz\pa_\bz mm^{-1}$ are given by \eqref{295}.
We find
\begin{subequations}\label{694}
    \begin{align}
        L_z(\lm)=( f_+(\lm)-\ri g_+(\lm)-\ri\mu R)(w -\ri\eta R_{m^{-1}}-\ri\mu R)^{-1}\pa_z mm^{-1},\\
        L_\bz(\lm)=( f_+(\lm)+\ri g_+(\lm)+\ri\mu R)(w +\ri\eta R_{m^{-1}}+\ri\mu R)^{-1}\pa_\bz mm^{-1}.
    \end{align}
\end{subequations}
 Note that the Lorentzian Lax pair \eqref{688} is the same as that described by the equations (8.86-89) of \cite{K20}, with the identifications
 $$b_l=\eta,\quad b_r=\mu, \quad a=w, \quad \hat f_\pm=f\mp g, \quad \xi =\frac{\eta^2-\mu^2-w^2+2w\eta\ri-2w^2\rho_-\frac{1-\lm}{1+\lm}}{\eta^2-\mu^2-w^2+2w\eta\ri+2w^2\rho_-\frac{1-\lm}{1+\lm}}.$$

\begin{remark}\small\label{rem:1206} The $\E$-Wick-rotated pair of operators $\E_\pm$ given by \eqref{1082} fits into the set up described in Section \ref{sec:4.2} because  $\E_+(w)=-\ri\E_-(-\ri w)$.  We can  check easily that
$$f_+(\lm,w)=\ri f_-(\ri\lm,-\ri w), \quad  g_+(\lm,w)=g_-(\ri\lm,-\ri w)$$
therefore indeed
$$O_+(\lm,w)=O_-(\Lambda(\lm,w),-\ri w), \quad \Lambda(\lm,w)=\ri\lm.$$
\end{remark}

 %%%%%%%%%%%%%%%%%%%%%%%%%%%%%%%%%%%%%%%%%%%
 \subsection{Lorentzian vs. Euclidean bi-Yang--Baxter renormalization}
 %%%%%%%%%%%%%%%%%%%%%%%%%%%%%%%%%%%%%%%%%%%
In the bi-Yang-Baxter   case, 
we look for a solution of the modified RG flow equation \eqref{784} of the form 
\be
\E_s(\tau)(a+\ri b)
=
-\frac{1}{w_s(\tau)}(\mu R+\ri \eta)a
+
\frac{w_s(\tau)}{\eta}
\left(
s\one+\frac{\mu^2}{w_s(\tau)^2}R^2
\right)b
+\ri\frac{\mu}{w_s(\tau)}Rb,  \quad a,b\in\K, \quad s=\pm,
\label{1914}
\ee
or, equivalently, of the form 
 \[
\E_s(\tau)=\frac{1}{w_s(\tau)} |V^a)(V^a|-sw_s(\tau) |T^a)(T^a|, \quad V^a=(\mu R+\ri\eta)T^a.
\]
Similarly as in Section \ref{sec:6.2}, we find that $\E_s(\tau)$
verifies the relation
\be
\frac{d\E_s}{d\tau}
=
\frac14(\E_s\M_s\E_s+s\M_s)
+\frac{2s\mu\eta}{w_s}
[\ad_{\rho^\vee},\E_s],
\label{1255}
\ee
 provided that $w_s$ satisfies the differential equations
\be
\frac{dw_s}{d\tau}
=\beta_s(w_s)=
-\frac{c_\K}{4w_s^2}
\left((\eta-\mu)^2-sw_s^2 \right)
\left((\eta+\mu)^2-sw_s^2\right).
\label{1264}
\ee
The symbol $\rho^\vee$ which appears in \eqref{1255}  is the standard dual Weyl vector
\[
\rho^\vee
=
\frac12\sum_{\alpha>0}\alpha^\vee
=
-\frac{\ri}{4}\sum_a[RT^a,T^a],
\]
where the sum runs over the positive roots and $\al^\vee$ stand for the coroots.

We conclude that the  family $\E_s(w)$ of the bi-Yang-Baxter operators \eqref{1082} is strictly renormalizable with the beta function given by
\eqref{1264}.
\begin{remark} We could have obtained the formula \eqref{1264} for the
beta functions without performing any computation. Indeed, the Lorentzian bi-Yang--Baxter beta-function  \eqref{1264} was already obtained   in
\cite{SST15} and,   because we are in the framework of  Section \ref{sec:5.3}, the analytic continuation
$$ \beta_+(w)=-\beta_-(-\ri w) $$  gives the beta function of the $\E$-Wick rotated Euclidean model.
 \end{remark}

 \medskip

It is worth observing that the fixed-point structure of the renormalization group flow is substantially richer in the Euclidean case than in the Lorentzian one. Indeed, the beta-function $\beta_-(w)$
has no finite real zero, since both factors are strictly positive for real nonzero \(w,\eta,\mu\):
\[
\beta_-(w)
=
-\frac{c_\K}{4w^2}
\left(w^2+(\eta-\mu)^2\right)
\left(w^2+(\eta+\mu)^2\right).
\]
On the other hand, in the Euclidean case   the beta-function $\beta_+(w)$ factorizes as
\[
\beta_+(w)
=
-\frac{c_\K}{4w^2}
\left(w^2-(\eta-\mu)^2\right)
\left(w^2-(\eta+\mu)^2\right),
\]
and it therefore admits four real fixed points $|\eta\pm\mu|$, $-|\eta\pm\mu|$.
%%%%%%%%%%%%%%%%%%%%%%%%%%%%%%%%%%%%%%%%%%%%%
\subsection{One-loop conformal deformation of the hyperbolic WZW model}
%%%%%%%%%%%%%%%%%%%%%%%%%%%%%%%%%%%%%%%%%%%%%
We have learned in Section  \ref{sec:6.2} that the second order Euclidean Yang-Baxter
$\sigma$-model \eqref{605} is completely singular at the Euclidean RG fixed point $w=\eta$. By saying "completely singular" we mean that it becomes singular in the conformal limit at every point on the simple compact group target $K$. 

\smallskip

Switching on the parameter $\mu$ changes the situation and e.g. for the one-loop conformal choice  $w=\eta+\mu$ the Euclidean  bi-Yang-Baxter model \eqref{1088} is  singular only at some particular points on the target $K$. For example, if $m$ 
 belongs to the Cartan torus then
\(\operatorname{Ad}_m\) commutes with the Yang--Baxter operator \(R\),
and hence \(R_{m^{-1}}=R\). The operator appearing in \eqref{1088}
therefore reduces to
\[
(\eta+\mu)(\one+\ri R),
\]
which fails to be invertible.  
This does not necessarily
invalidate the one-loop conformal choice $w=\eta+\mu$ as a two-dimensional $\sigma$-model on the target $K$; it rather indicates
that this $\sigma$-model geometry is not regular on the
whole group manifold.  

\smallskip 

This situation is reminiscent of gauged WZW models, where the elimination of
the auxiliary gauge fields may break down at target-space points having a
non-trivial stabilizer under the gauge-group action even though the
underlying gauged formulation is well-defined.
Similarly, for the one-loop conformal choice $w=\eta+\mu\neq 0$ the Euclidean  bi-Yang--Baxter \(\E_+\)-operator itself
remains a perfectly well-defined linear operator on the Drinfeld double \(\D\). What degenerates is not the
\(\E\)-model data, but rather the particular passage to the second-order sigma-model
description, where one has to invert an operator which becomes singular along the
Cartan torus.

\smallskip

It actually  happens that the Poisson-Lie dual passage to the second order description  for the one-loop conformal choice $  w=\eta+\mu$  gives perfectly regular $\sigma$-model living on the coset space target $K^\bC/K$ provided $-\eta<\mu\leq 0$.
The space of cosets $K^\bC/K$ is canonically identified with the space $K^{+}$ formed by  positive definite Hermitian elements of $K^\bC$. The dual action is obtained by inserting $w=\eta +\mu$ into \eqref{eyc+} where $P\in K^{+}$
\begin{multline}
\tilde S_+(P)
=
-\frac{1}{8\eta}\oint dt\,d\sigma\;
\tr\left[
\left(
\frac{\one+\frac{\mu}{2\eta}(\one-\ri R)(\one+
 \Ad_P)
}{
\one+\frac{\mu}{2\eta}(\one-\ri R)(\one-
 \Ad_P)
}
\right)
(P^{-1}\partial_z P)(P^{-1}\partial_{\bar z}P)
\right]
\\
-\frac{\ri}{8\eta}
\int d\sigma\,d^{-1}
\tr\left(
P^{-1}dP,
[
P^{-1}P',
P^{-1}dP
]
\right).
\label{eyc++}
\end{multline}

Let us show that the background \eqref{eyc++} is indeed regular for $-\eta<\mu\leq 0$.  Equivalently, this means that the operator 
$$
D_P:=\one+\frac{\mu}{2\eta}(\one-\ri R)(\one-
 \Ad_P)
$$
is invertible for that range of the parameters.

\smallskip

Consider the natural positive-definite Hermitian scalar product $(.,.)_h$ on
$\K^{\bC}$ defined by 
$$(x,y)_h=\tr(x^\dagger y), \quad x,y\in\K^\bC.$$
We denote by $T^*$ the adjoint of an operator $T$ with respect to this
Hermitian product. We find
$$(\one-\ri R)^*=(\one-\ri R), \quad \Ad_P^*=\Ad_P.$$ Moreover, $P=e^{\ri a}$ for some $a\in\K$, so that
\[
\Ad_P=e^{\ad_{\ri a}}>0.
\]
Here the symbol $>$ means that the self-adjoint operator $\Ad_P$ is strictly positive definite.

Similarly, since 
$\operatorname{Spec}(\one-\ri R)\subset\{0,1,2\}$, we have
$$\one-\ri R\geq 0.$$

\smallskip

Set
$$
Q:=\one+q(\one-\ri R),$$
where 
$$ q:=\frac{\mu}{2\eta}.$$
Because $-\frac12<q\leq0$ and $0\leq \one-\ri R\leq2\one$, the operator $Q$ is
strictly positive and therefore invertible. 
Set 
$$
M:=qQ^{-1}(\one-\ri R),
$$
so that 
we can   factorize $D_P$ as
$$
D_P=Q(\one-M\Ad_P).
$$
Furthermore, $M$ is
self-adjoint and negative semi-definite. 
Although $M\Ad_P$ need not itself be self-adjoint, it is similar to the
self-adjoint operator
$$
\Ad_P^{1/2}M\Ad_P^{1/2},
$$
 which is negative semi-definite and hence
$$
\operatorname{Spec}(M\Ad_P)\subset(-\infty,0].
$$
In particular, $1$ cannot be an eigenvalue of $M\Ad_P$. Therefore
$\one-M\Ad_P$, and consequently $D_P$, is invertible for every
$P\in K^+$.

It follows that
$$
D_P^{-1}(\one+\frac{\mu}{2\eta}(\one-\ri R)(\one+
 \Ad_P))$$
is a smooth operator-valued function on the whole target
$K^+$. The Wess--Zumino three-form in \eqref{eyc++} is manifestly
smooth as well, so no pole can occur in the background fields.

\smallskip 

We know prove that the background metric on $K^+$ corresponding to the $\sigma$-model action \eqref{eyc++} is not only everywhere regular but also nowhere degenerate.

We introduce the notation
$$ N_P= \one+\frac{\mu}{2\eta}(\one-\ri R)(\one+
 \Ad_P),$$
 so that 
 the operator entering the action \eqref{eyc++} is
\[
O_P=D_P^{-1}N_P.
\]
 We also set 
\[
L:=\one-\ri R,
\qquad
\widehat L:=\one+\ri R,
\qquad
A:=\Ad_P,
\]
and rewrite the operators $D_P,N_P$ as
\[
D:=\one+qL(\one-A),
\qquad
N:=\one+qL(\one+A).
\]

Let $t$ denote the transpose with respect to the complex bilinear trace
form
\[
b(X,Y):=\tr(XY).
\]
Since
\[
R^t=-R,
\qquad
A^t=A^{-1},
\]
we have
\[
L^t=\widehat L
\]
and hence
\[
D^t=\one+q(\one-A^{-1})\widehat L,
\qquad
N^t=\one+q(\one+A^{-1})\widehat L.
\]
The symmetric operator determining the target-space metric is therefore
\be 
G
:=\frac12\left(O+O^t\right)
=
\frac12D^{-1}
\left(
ND^t+DN^t
\right)
(D^t)^{-1}.
\label{metricoperator}
\ee

A direct expansion gives
 $$
ND^t+DN^t=2(1+2q)\one.$$
 
Substitution into \eqref{metricoperator} yields the factorization
\be
G=(1+2q)D^{-1}(D^t)^{-1}.
\label{metricfactorization}
\ee
 Since $D$ is everywhere invertible in the parameter
range $1+2q>0$, $D^t$ is everywhere invertible as well.
Consequently, $G$ is an invertible operator on $\K^{\bC}$.

\smallskip

It remains to verify that this implies non-degeneracy on the real tangent
space of $K^+$. Note that the tangent current
\[
X=P^{-1}\delta P
\]
satisfies
\[
X^\dagger
=
\delta P\,P^{-1}
=
PXP^{-1}
=
AX.
\]
The real tangent space at $P$ can therefore be identified with
\[
\mathfrak p_P
:=
\left\{
X\in\K^{\bC}
\mathrel{;}
X^\dagger=AX
\right\}.
\]

For an operator $T$ on $\K^{\bC}$, define $T^\sharp$ by
$$T^\sharp X=
(TX^\dagger)^\dagger.
$$
We then have
\[
A^\sharp=A^{-1},
\qquad
L^\sharp=\widehat L.
\]
It follows that
\[
D^\sharp
=
\one+q\widehat L(\one-A^{-1})
\]
and
\[
(D^t)^\sharp
=
\one+q(\one-A)L
\]
Using
\[
L+\widehat L=2\one,
\qquad
L\widehat L=\widehat L L,
\]
one verifies directly that
\be
(D^t)^\sharp D^\sharp=AD^tDA^{-1}.
\label{realityidentity}
\ee

Equation \eqref{metricfactorization} gives
\[
G^{-1}
=
\frac{1}{1+2q}D^tD.
\]
Consequently, using \eqref{realityidentity}, we find
\[
(G^{-1})^\sharp
=
\frac{1}{1+2q}(D^t)^\sharp D^\sharp
=
AG^{-1}A^{-1},
\]
and therefore
\[
G^\sharp=AGA^{-1}.
\]
If $X\in\mathfrak p_P$, it follows that
\[
(GX)^\dagger
=
G^\sharp X^\dagger
=
AGA^{-1}AX
=
AGX.
\]
Hence
\[
G\mathfrak p_P\subset\mathfrak p_P.
\]
Since $G$ is invertible, its restriction to the finite-dimensional real
space $\mathfrak p_P$ is an automorphism.

Finally, the restriction of the trace form $b$ to $\mathfrak p_P$ is
non-degenerate. Indeed, setting
\[
H:=\Ad_{P^{1/2}}X,
\]
the condition $X^\dagger=AX$ gives
\[
H^\dagger=H.
\]
Thus
\[
\mathfrak p_P
=
\Ad_{P^{-1/2}}(\ri\K).
\]
By invariance of the trace,
\[
b(X,Y)
=
b\left(
\Ad_{P^{1/2}}X,
\Ad_{P^{1/2}}Y
\right),
\]
and the restriction of $b$ to $\ri\K$ is definite and, in particular,
non-degenerate. Therefore $b|_{\mathfrak p_P}$ is non-degenerate.

Suppose that $X\in\mathfrak p_P$ belongs to the kernel of the
target-space metric. This means that
\[
b(X,GY)=0
\qquad
\text{for every }Y\in\mathfrak p_P.
\]
Since $G$ maps $\mathfrak p_P$ onto itself, this implies
\[
b(X,Z)=0
\qquad
\text{for every }Z\in\mathfrak p_P.
\]
The non-degeneracy of $b|_{\mathfrak p_P}$ then gives $X=0$.
Therefore the target-space metric is non-degenerate at every point at
which $D$ is invertible. In particular, for
\[
\eta>0,
\qquad
-\eta<\mu\leq0,
\]
the operator $D$ is invertible for every $P\in K^+$.
The metric is consequently non-degenerate everywhere. Since it is
definite at $q=0$, its signature cannot change throughout this connected
parameter range. Hence the background \eqref{eyc++} is everywhere
regular and possesses a definite target-space metric.

\smallskip

 We finish this section by remarking that, up to the overall $\frac1\eta$ normalization,    the action \eqref{eyc++} describes the one-parameter deformation of the hyperbolic WZW model.
Since  this deformation is  one-loop conformal there should be a good chance to quantize the model exactly for each value of the deformation parameter $q$. The probabilistic methods of \cite{GKR25} look particularly promising in this respect.
 
 %%%%%%%%%%%%%%%%%%%%%%%%%%%%%%%%%%%%%%%%%%
\section{Conclusions and outlook}
%%%%%%%%%%%%%%%%%%%%%%%%%%%%%%%%%%%%%%%%%%

In this paper, we developed a systematic framework for the study of what we have called
Euclidean $\E$-models, namely first-order systems on Drinfeld doubles for which the  self-adjoint operator $\E$ squares to $-\one$ rather than to $\one$.
Although this modification is formally simple, it leads to a substantially different geometric and dynamical picture.
Most importantly, the associated second-order $\sigma$-models naturally live on a Euclidean world-sheet and possess real Euclidean actions.

\medskip

Our first goal was to formulate the Euclidean theory in a language as parallel as possible to the standard Lorentzian $\E$-model formalism.
This led to Euclidean counterparts of several basic constructions:
the first-order equations of motion, the passage to the second-order $\sigma$-model, and the corresponding form of Poisson--Lie T-duality.
In particular, for perfect Drinfeld doubles, the dual $\sigma$-models on mutually dual Poisson--Lie groups admit an explicit description entirely analogous to the Lorentzian one, with the differences encoded in a controlled and conceptually transparent way.

\smallskip

An important advantage of the Euclidean $\E$-models framework is that it places on a structural basis various analytic continuations studied previously in the literature \cite{BKL18}, \cite{LS18} which, at the level of individual $\sigma$-model actions, might otherwise appear as  somewhat ad hoc prescriptions.  The Euclidean $\E$-model formalism  permitted us to justify that the analytic continuations of mutually Poisson--Lie dual Lorentzian models remain mutually Poisson--Lie dual on the Euclidean side. 

\smallskip 

Our second goal was to find a natural procedure associating a Euclidean model with a given Lorentzian one. For perfect Drinfeld doubles, this was achieved by the $\E$-Wick rotation, which assigns to every Lorentzian $\E$-model a canonical Euclidean partner with a real Euclidean action. This operation is distinct from the standard Wick rotation of the corresponding second-order action: it acts on the underlying first-order data and transforms both the symmetric and the antisymmetric parts of the target-space background in a nontrivial way.

\smallskip 

The examples studied in this paper provide substantial evidence that the $\E$-Wick rotation is compatible not only with Poisson--Lie duality but also with integrability and renormalization. Within the particular ansatz considered here, Lorentzian Lax representations give rise, by analytic continuation, to Euclidean Lax representations, and the Euclidean beta functions are obtained from their Lorentzian counterparts in the same way. This compatibility is realized explicitly by the Euclidean bi-Yang--Baxter model.

\smallskip 

The scope of this result must nevertheless be stated carefully. We do not at present have a general argument showing that every solution of the sufficient Lorentzian integrability condition automatically produces a solution of its Euclidean counterpart under the $\E$-Wick rotation. Nor have we proved that the beta functions of an arbitrary strictly renormalizable Lorentzian family are necessarily transported into those of the associated Euclidean family. The compatibility with integrability and renormalization has so far been established only within the particular operator ansatz investigated in this work. Thus, although our results strongly support the naturalness of the $\E$-Wick rotation, its general status in this respect remains to be clarified.

 \smallskip 

 The Euclidean bi-Yang--Baxter model also illustrates the usefulness of combining the intrinsic Euclidean formalism with Poisson--Lie duality. We identified a one-parameter family of one-loop RG fixed points for which the compact-target geometry develops singularities on the maximal torus of $K$, whereas the dual model on $K^{\bC}/K$ possesses an everywhere regular background. The latter provides an integrable one-parameter family of one-loop conformal deformations of the non-unitary hyperbolic Wess--Zumino--Witten model. This example shows that a Euclidean theory which appears singular in one Poisson--Lie frame may admit a smooth and geometrically meaningful description in the dual frame.

\medskip

There are several natural directions for further investigation.

\medskip

First, it would be important to enlarge the class of explicit integrable Euclidean $\E$-models. In the present paper, we treated Yang--Baxter and bi-Yang--Baxter type examples, but more general integrable deformations should also admit Euclidean counterparts. In particular, it would be interesting to investigate elliptic integrable models from the perspective of Euclidean $\E$-models. The rational and trigonometric cases already fit naturally into the $\E$-model framework, and one may expect the elliptic case to exhibit a richer Euclidean geometry, possibly accompanied by new forms of Poisson--Lie duality and Lax representations.
 
\smallskip

Second, Euclidean analogues of dressing cosets or degenerate $\E$-models should exist as well. Developing these extensions may substantially enlarge the class of non-unitary $\sigma$-models with real Euclidean actions accessible within the present framework.

\smallskip

Third, it would be desirable to relate our treatment of $\E$-models to the perspective offered by affine Gaudin models and four-dimensional Chern--Simons theory \cite{LV21}. In the latter construction, the integrable models need not be Lorentzian and depend on free parameters $\epsilon_i$ which, in the $\E$-model interpretation, become the eigenvalues of the operator $\E$. In particular, the models are relativistic when $\epsilon_i^2=1$, in agreement with the Lorentzian $\E$-model condition $\E^2=\one$. Similarly, motivated by \cite{CY19}, it was proposed in \cite{DLMV20} that the choice $\epsilon_i^2=-1$ leads to integrable Euclidean models. This is consistent with the Euclidean $\E$-model condition $\E^2=-\one$ and suggests that the construction may produce a large class of new integrable Euclidean $\E$-models. It would be valuable to understand systematically how these models fit into the framework developed here.

\medskip

Finally, the most important long-term direction concerns quantization. Recent progress in the probabilistic construction of non-unitary theories with real Euclidean actions suggests that Euclidean $\E$-models may provide a natural setting in which to reconsider the quantum meaning of Poisson--Lie duality and integrability. It would be particularly interesting to understand whether the classical structures developed here admit quantum counterparts formulated directly at the level of the Euclidean path integral. The integrable one-parameter family of one-loop conformal deformations of the hyperbolic Wess--Zumino--Witten model constructed in this paper provides a concrete testing ground for this question.

 \appendix
 \section{Appendix}

\begin{theorem}
 
 Consider a  Drinfeld double $\D$   and a  linear operator $\E:\D\to \D$   such that  for all $x,y\in \D$ it holds $(\E x,y)_\D=(x,\E y)_\D$ and such that  $\E^2=-\one$. It follows that the signature of the symmetric bilinear form $(x,y)_\E:=(x,\E y)_\D$ is $(n,n)$, that is it  is the same as the signature of $(.,.)_\D$.

\end{theorem}

\begin{proof} Let $T_1,\dots,T_n,t^1,\dots,t^n$
be a  basis of $\D$ such that 
\[
(T_i,t^j)_\D=\delta_{i}^{\ j}, \qquad (T_i,T_j)_\D=(t^i,t^j)_\D=0.
\]
If, by an abuse of notation, we denote by the same symbols $x,y$  the column $2n$-vectors corresponding to the elements $x,y\in\D$ in the basis $T_i,t^i$, then the bilinear form $(.,.)_\D$ can be expressed as 
\[
(x,y)_\D = x^{\mathsf T} J y,
\]
where the Gram matrix $J$ of $(\cdot,\cdot)_\D$ is
\be 
J=\begin{pmatrix}
0 & I_n\\
I_n & 0
\end{pmatrix}.\label{300}
\ee
In particular, it follows from \eqref{300} that  $(\cdot,\cdot)_\D$ has manifestly the signature $(n,n)$.

\medskip

If, again by an abuse of notation, we denote by the symbol $\E$ the matrix of the
operator $\E$ in the basis $T_i,t^i$, we
find
\[
(x,\E y)_\D = x^{\mathsf T} J\E y,
\]
The Gram matrix of $(\cdot,\E\cdot)_\D$ is therefore
\[
S := J \E.
\]
and the condition $(\E x,y)_\D=(x,\E y)_\D$
 becomes
\[
\E^{\mathsf T}J = J \E.
\]
We compute
\[
S^{\mathsf T} = (J \E)^{\mathsf T} = \E^{\mathsf T} J = J \E = S.
\]
 Since $\E^2=-I$, $\E$ is invertible, hence $S$ is invertible and
$(\cdot,\E\cdot)_\D$ is nondegenerate.

\medskip

Using $\E = J S$ (since $J^2=I$), the condition $\E^2=-I$ becomes
\[
(JS)(JS) = -I,
\]
hence
\[
S^{-1} = - J S J.
\]
 We observe that $J S J$ is congruent to $S$, hence has the same signature as $S$.
Multiplying by $-1$ swaps the numbers of positive and negative eigenvalues.
 On the other hand,  $S$ and $S^{-1}$ have the same signature.
This is only possible if the numbers of positive and negative eigenvalues are equal.
Therefore
\[
\mathrm{sign}(\cdot,\E\cdot)_\D = (n,n).
\]
 
\end{proof}
  
\begin{theorem}
   
{\it Let $\mathcal E:\mathcal D\to\mathcal D$ be a linear operator  such that $\mathcal E^2=\one$ and the bilinear form $g_{\mathcal E}(x,y):=(x,\mathcal E y)_{\mathcal D}$ is symmetric and strictly positive definite.
Assume  that $W\subset\mathcal D$ is an $n$-dimensional subspace which is
 maximally isotropic with respect to $(\cdot,\cdot)_{\mathcal D}$, i.e.
 $(W,W)_{\mathcal D}=0$.
Then it holds
$
\mathcal E(W)\cap W=\{0\}$.}
\end{theorem}

\begin{proof}
Suppose that  $w\in \mathcal E(W)\cap W$.
Then $w=\mathcal E v$ for some $v\in W$. Using symmetry of $\mathcal E$ and the involution property, we compute
\[
g_{\mathcal E}(w,w)
=(w,\mathcal E w)_{\mathcal D}
=(\mathcal E v,\mathcal E^2 v)_{\mathcal D}
=(\mathcal E v,v)_{\mathcal D}
=(v,\mathcal E v)_{\mathcal D}
=(v,w)_{\mathcal D}.
\]
Since both $v$ and $w$ belong to $W$ and $W$ is totally isotropic with respect to
$(\cdot,\cdot)_{\mathcal D}$, we have
\[
(v,w)_{\mathcal D}=0.
\]
Thus
\[
g_{\mathcal E}(w,w)=0,
\]
which contradicts the positive definiteness of $g_{\mathcal E}$ unless $w=0$.
Therefore $\mathcal E(W)\cap W=\{0\}$.
\end{proof}
 
\begin{theorem}

{\it Let $\mathcal E:\mathcal D\to\mathcal D$ be a linear operator  such that $\mathcal E^2=-\one$ and the bilinear form $g_{\mathcal E}(x,y):=(x,\mathcal E y)_{\mathcal D}$ is symmetric (it is then, following Theorem A.1, non-degenerate and has signature $(n,n)$).
Assume that $W\subset\mathcal D$ is an $n$-dimensional maximally isotropic subspace such that the bilinear form $g_{\mathcal E}$ restricted on $W$
is non-degenerate.
Then it holds
$
\mathcal E(W)\cap W=\{0\}$.}
\end{theorem}
\begin{proof} Assume for contradiction that
\[
W\cap \E(W)\neq\{0\}.
\]
Choose a nonzero vector $w\in W\cap \E(W)$. Then there exists $u\in W$ such that $w=\E u$.
Note that $u$ and $\E u$ are linearly independent because $\E$ does not have eigenvectors. Moreover, since $\E^2=-\one$,
the plane
$U:=\operatorname{span}\{u,\E u\}\subset W$ is  
  $\E$-invariant subspace of $W$.

  \medskip

Let us finally show that $U\subset {\rm ker}(g_\E)$.
Indeed, let  $x\in U$ and $y\in W$ be arbitrary.
Using symmetry of $\E$, we compute
\[
g_\E(x,y)=(x,\E y)=(\E x,y).
\]
Since $x\in U$ and $U$ is $\E$-invariant, we have $\E x\in U\subset W$.
Because $W$ is maximally  isotropic for $(\cdot,\cdot)$, it follows that
$$
(\E x,y)=g_\E(x,y)=0,\qquad \forall x\in U,\ \forall y\in W.$$
It follows that every $x\in U$ lies in the kernel of the restriction of the bilinear
form $g_\E$  to the subspace  $W$
Since $U$ is $2$-dimensional and nonzero, this contradicts the assumption that
this restriction  is nondegenerate. Said in other words, our assumption $W\cap \E(W)\neq\{0\}$ must be false, and we conclude that
\[
W\cap \E(W)=\{0\}.
\]
 
 \end{proof}
 \begin{theorem}
Let $\D=\ti\K\oplus \K$ with the canonical bilinear form
\begin{equation}
\big((\ti\alpha,a),(\ti\beta,b)\big)_\D
=(\ti\alpha,b)_\D+(\ti\beta,a)_\D .
\end{equation}
Let $\E_-:\D\to\D$ be a linear operator satisfying
\begin{equation}
\E_-^2=1
\end{equation}
and suppose that the bilinear form
\begin{equation}
B(x,y):=(x,\E_-y)_\D
\end{equation}
is symmetric and strictly positive definite.

Then $\E_-$ necessarily has the form
\begin{equation}
\E_-=
\begin{pmatrix}
-\ti S_-^{-1}\ti A_- & \ti S_-^{-1}\\
S_-^{-1} & -S_-^{-1}A_-
\end{pmatrix},
\end{equation}
where $S_-,\ti S_-$ are invertible symmetric operators and
$A_-,\ti A_-$ are antisymmetric operators.
\end{theorem}
\begin{proof}

Write $\E_-$ in block form with respect to the decomposition
$\D=\ti\K\oplus\K$:
\begin{equation}
\E_-=
\begin{pmatrix}
X & Y\\
Z & W
\end{pmatrix},
\qquad
X:\ti\K\to\ti\K,\;
Y:\K\to\ti\K,\;
Z:\ti\K\to\K,\;
W:\K\to\K .
\end{equation}

\medskip

\textbf{1. Positivity of $B$ implies that $Y$ and $Z$ are symmetric and invertible.}

For $x=(0,a)\in\D$ we obtain
\begin{equation}
B(x,x)=(x,\E_-x)_\D=(a,Ya)_\D .
\end{equation}
Since $B$ is positive definite, $(a,Ya)_\D>0$ for all $a\neq0$.
For $x=(0,a)$ and $y=(0,b)$ one has
\begin{equation}
B(x,y)=((0,a),\E_-(0,b))_\D=(a,Yb)_\D .
\end{equation}
Hence the restriction of $B$ to the subspace $\K\subset \D$ is the bilinear form
\begin{equation}
g_Y(a,b):=(a,Yb)_\D.
\end{equation}
Since $B$ is symmetric, $g_Y$ is symmetric, therefore
\begin{equation}
(a,Yb)_\D=(b,Ya)_\D=(Ya,b)_\D ,
\end{equation}
which means that $Y$ is symmetric. Since $B$ is strictly  positive definite, for every
$a\neq 0$,
\begin{equation}
(a,Ya)_\D=B((0,a),(0,a))>0 .
\end{equation}
Thus $g_Y$ is also strictly positive definite, hence nondegenerate. If $Ya=0$, then
\begin{equation}
(a,Yb)_\D=(Ya,b)_\D=0
\qquad \forall b\in\K ,
\end{equation}
so in particular $(a,Ya)_\D=0$, which implies $a=0$. Therefore $Y$ is injective,
and in finite dimension it is invertible.

Similarly, for $x=(\ti\alpha,0)$ we obtain
\begin{equation}
B(x,x)=(\ti\alpha,Z\ti\alpha)_\D>0 ,
\end{equation}
which implies that $Z$ is symmetric and invertible.

Define therefore
\begin{equation}
\ti S_-:=Y^{-1},\qquad S_-:=Z^{-1}.
\end{equation}
Then $\ti S_-$ and $S_-$ are invertible and symmetric.

\medskip

\textbf{2. Definition of $A_-$ and $\ti A_-$.}

Define linear operators
\begin{equation}
A_-:=-S_-W=-Z^{-1}W,\qquad
\ti A_-:=-\ti S_-X=-Y^{-1}X .
\end{equation}
Equivalently,
\begin{equation}
W=-S_-^{-1}A_-,\qquad
X=-\ti S_-^{-1}\ti A_- .
\end{equation}

Substituting these expressions into the block form of $\E_-$ gives
\begin{equation}
\E_-=
\begin{pmatrix}
-\ti S_-^{-1}\ti A_- & \ti S_-^{-1}\\
S_-^{-1} & -S_-^{-1}A_-
\end{pmatrix}.
\end{equation}

\textbf{3. Antisymmetry of $A_-$ and $\ti A_-$.}

First we use the symmetry of the bilinear form
\begin{equation}
B(x,y):=(x,\E_-y)_\D .
\end{equation}

For $x=(\ti\alpha,0)$ and $y=(0,a)$ we compute
\begin{equation}
\E_-(0,a)=(Ya,Wa), \qquad
\E_-(\ti\alpha,0)=(X\ti\alpha,Z\ti\alpha).
\end{equation}
Hence
\begin{align}
B(x,y)&=((\ti\alpha,0),(Ya,Wa))_\D=(\ti\alpha,Wa)_\D,\\
B(y,x)&=((0,a),(X\ti\alpha,Z\ti\alpha))_\D=(a,X\ti\alpha)_\D.
\end{align}
Since $B$ is symmetric we obtain
\begin{equation}
(\ti\alpha,Wa)_\D=(a,X\ti\alpha)_\D
\qquad\forall\,a,\ti\alpha,
\end{equation}
hence
\begin{equation}
W^t=X .
\end{equation}

Recall that
\begin{equation}
A_-:=-Z^{-1}W .
\end{equation}
Taking the transpose gives
\begin{equation}
A_-^t=-W^t(Z^{-1})^t= -XZ^{-1}.
\end{equation}

From the relation $\E_-^2=1$ we have the block identity
\begin{equation}
ZX+WZ=0 .
\end{equation}
Multiplying by $Z^{-1}$ on the left and right yields
\begin{equation}
XZ^{-1}+Z^{-1}W=0 .
\end{equation}
Hence
\begin{equation}
XZ^{-1}=-Z^{-1}W,
\end{equation}
or, equivalently
\begin{equation}
A_-^t =-A_- .
\end{equation}
Thus $A_-$ is antisymmetric.
\medskip

The same argument applied to the block relation
\begin{equation}
XY+YW=0
\end{equation}
shows that
\begin{equation}
\ti A_-:=-Y^{-1}X
\end{equation}
satisfies
\begin{equation}
\ti A_-^t=-\ti A_- .
\end{equation}
Hence $\ti A_-$ is antisymmetric as well.

\medskip

Therefore $\E_-$ has the claimed form with
$S_-,\ti S_-$ symmetric and invertible and
$A_-,\ti A_-$ antisymmetric.
\end{proof}

 \noindent {\bf Two ways of writing the Lorentzian  PL dual of the Yang-Baxter model}

 \medskip 
 
We use the parametrization
\[
b=\widetilde m
=
\begin{pmatrix}
e^{\eta\varphi}
&
2\eta e^{-\eta\varphi}(u_1-i u_2)
\\
0&e^{-\eta\varphi}
\end{pmatrix},
\qquad
P=bb^\dagger .
\]

Then
\[
\partial_\pm b\,b^{-1}
=
Y_\pm^1\tau_1+Y_\pm^2\tau_2+Y_\pm^3\tau_3,
\]
where
\[
Y_\pm^1=\partial_\pm u_1-2\eta u_1\partial_\pm\varphi,
\qquad
Y_\pm^2=\partial_\pm u_2-2\eta u_2\partial_\pm\varphi,
\qquad
Y_\pm^3=\partial_\pm\varphi .
\]

The dual Poisson tensor is
\[
\widetilde\Pi(b)=
\begin{pmatrix}
0&\rho&-u_2\\
-\rho&0&u_1\\
u_2&-u_1&0
\end{pmatrix},
\]
with
\[
\rho
=
\eta(u_1^2+u_2^2)
+
\frac{e^{4\eta\varphi}-1}{4\eta}.
\]

Let
\[
B:=\alpha^2+\eta^2,
\qquad
a:=\frac{\alpha}{2B},
\qquad
c:=\frac{1}{2\alpha},
\qquad
d:=\rho+\frac{\eta}{2B}.
\]
Then
\[
S_- - A_-+\widetilde\Pi(b)
=
K
=
\begin{pmatrix}
a&d&-u_2\\
-d&a&u_1\\
u_2&-u_1&c
\end{pmatrix}.
\]

The determinant of \(K\) is
\[
\Delta
=
a^2c+a(u_1^2+u_2^2)+cd^2 .
\]
Its inverse is
\[
K^{-1}
=
\frac1{\Delta}
\begin{pmatrix}
ac+u_1^2&-cd+u_1u_2&au_2+du_1\\
cd+u_1u_2&ac+u_2^2&-au_1+du_2\\
-au_2+du_1&au_1+du_2&a^2+d^2
\end{pmatrix}.
\]

Therefore the Poisson--Lie action is
\[
\widetilde{\mathcal S}_-(b)
=
\frac12\int dt\,d\sigma\;
Y_+^T K^{-1}Y_- .
\]

Equivalently,
\[
\begin{aligned}
\widetilde{\mathcal L}_-
=
\frac{1}{2\Delta}\Big[
&
(ac+u_1^2)Y_+^1Y_-^1
+
(ac+u_2^2)Y_+^2Y_-^2
+
(a^2+d^2)Y_+^3Y_-^3
\\
&+
(-cd+u_1u_2)Y_+^1Y_-^2
+
(cd+u_1u_2)Y_+^2Y_-^1
\\
&+
(au_2+du_1)Y_+^1Y_-^3
+
(-au_2+du_1)Y_+^3Y_-^1
\\
&+
(-au_1+du_2)Y_+^2Y_-^3
+
(au_1+du_2)Y_+^3Y_-^2
\Big].
\end{aligned}
\]

On the other hand, consider the action
\[
\begin{aligned}
S_-(P)
&=
-\frac{i}{8\eta}
\int dt\,d\sigma\;
\operatorname{tr}
\left(
\frac{
\alpha-i\eta+(\alpha+i\eta)\operatorname{Ad}_P
}{
\alpha-i\eta-(\alpha+i\eta)\operatorname{Ad}_P
}
(P^{-1}\partial_+P)(P^{-1}\partial_-P)
\right)
\\
&\quad
-\frac{i}{8\eta}
\int d\sigma\,d^{-1}
\operatorname{tr}
\left(
P^{-1}dP,
[
P^{-1}P',
P^{-1}dP
]
\right).
\end{aligned}
\]

Using
\[
\int d\sigma\,d^{-1}
\operatorname{tr}
\left(
P^{-1}dP,
[
P^{-1}P',
P^{-1}dP
]
\right)
=
-
\int dt\,d\sigma\;
\operatorname{tr}
\Big[
b^{-1}\partial_+b\,
\partial_-b^\dagger(b^\dagger)^{-1}
-
b^{-1}\partial_-b\,
\partial_+b^\dagger(b^\dagger)^{-1}
\Big],
\]
the WZ contribution can be written locally on the \(P=bb^\dagger\) patch.

For the general \(AN\) element \(b\) above, this local primitive is nonzero. Explicitly,
\[
\begin{aligned}
&\operatorname{tr}
\Big[
b^{-1}\partial_+b\,
\partial_-b^\dagger(b^\dagger)^{-1}
-
b^{-1}\partial_-b\,
\partial_+b^\dagger(b^\dagger)^{-1}
\Big]
\\
&\qquad
=
8i\eta^2e^{-4\eta\varphi}
\left(
\partial_+u_1\,\partial_-u_2
-
\partial_-u_1\,\partial_+u_2
\right).
\end{aligned}
\]

Thus the WZ contribution to \(S_-(P)\) is
\[
-\eta e^{-4\eta\varphi}
\left(
\partial_+u_1\,\partial_-u_2
-
\partial_-u_1\,\partial_+u_2
\right).
\]

Substituting
\[
P=bb^\dagger
\]
into the Cayley-transform term and adding the WZ primitive gives
\[
\mathcal L_P
=
\frac12\,Y_+^T K^{-1}Y_- .
\]

Hence
\[
S_-(P)=\widetilde{\mathcal S}_-(b).
\]

\acknowledgments

%\section*{Acknowledgements}

\smallskip

I am indebted to Dan Thompson, Kostas Sfetsos and Kostas Siampos for working out the crucial formula \eqref{fl-} which characterizes the RG behaviour of the Euclidean $\E$-models.

\end{document}